\documentclass[11pt]{article}

\usepackage{amsthm,amsmath}
\RequirePackage[numbers]{natbib}
\RequirePackage[colorlinks,citecolor=blue,urlcolor=blue]{hyperref}

\usepackage[utf8]{inputenc}
\usepackage{makeidx}
\usepackage{algorithm}
\usepackage{float}
\usepackage{algpseudocode}
\usepackage{amsfonts}
\usepackage{amssymb}
\usepackage{gensymb}
\usepackage{caption}
\usepackage[labelformat=simple]{subcaption}
\usepackage{bbm}
\usepackage{lipsum}
\usepackage{enumitem}
\usepackage{hyperref}
\usepackage[noconfigs,english]{babel} 
\usepackage{bm} 
\usepackage[T1]{fontenc}	 
\usepackage{lmodern}
\usepackage{amsthm} 
\usepackage{graphicx} 
\usepackage{booktabs} 
\usepackage{multicol} 
\usepackage{tocstyle} 
\usepackage[toc,page]{appendix}

\usepackage{array}
\usepackage{tabularx}

\numberwithin{equation}{section}

\theoremstyle{definition}

\theoremstyle{plain}
\newtheorem{theorem}{Theorem}[section]
\newtheorem{lemma}{Lemma}[section]

\newtheorem{proposition}{Proposition}[section]

\newcommand{\numbereqn}{\addtocounter{equation}{1}\tag{\theequation}} 

\allowdisplaybreaks

\DeclareMathOperator*{\argmin}{arg\,min}
\algnewcommand{\Inputs}[1]{%
  \State \textbf{Inputs:}
  \Statex \hspace*{\algorithmicindent}\parbox[t]{.8\linewidth}{\raggedright #1}
}
\algnewcommand{\Initialize}[1]{%
  \State \textbf{Initialize:}
  \Statex \hspace*{\algorithmicindent}\parbox[t]{.8\linewidth}{\raggedright #1}
}

\begin{document}


%
%
%
%
%
\hypersetup{linkcolor=blue}

\date{\today}

\author{Ray Bai \thanks{Ray Bai (email: raybai07@ufl.edu) is Postdoctoral Researcher, Department of Biostatistics, Epidemiology, and Informatics, University of Pennsylvania.} \\
Malay Ghosh  \thanks{Malay Ghosh  (email: ghoshm@ufl.edu) is Distinguished Professor, Department of Statistics, University of Florida.} \\ }

\title{On the Beta Prime Prior for Scale Parameters in High-Dimensional Bayesian Regression Models \thanks{Keywords and phrases:
{empirical Bayes},
{high-dimensional data},
{linear regression},
{shrinkage estimation},
{scale mixtures of normal distributions},
{posterior contraction}
}
}

\maketitle
\begin{abstract}
We study high-dimensional Bayesian linear regression with a general beta prime distribution for the scale parameter. Under the assumption of sparsity, we show that appropriate selection of the hyperparameters in the beta prime prior leads to the (near) minimax posterior contraction rate when $p \gg n$.  For finite samples, we propose a data-adaptive method for estimating the hyperparameters based on marginal maximum likelihood (MML). This enables our prior to adapt to both sparse and dense settings, and under our proposed empirical Bayes procedure, the MML estimates are never at risk of collapsing to zero. We derive efficient Monte Carlo EM and variational EM algorithms for implementing our model, which are available in the \textsf{R} package \texttt{NormalBetaPrime}. Simulations and analysis of a gene expression data set illustrate our model's self-adaptivity to varying levels of sparsity and signal strengths.
 \end{abstract}

%
%
%
%

\section{Introduction} \label{intro}

\subsection{Background} \label{background}
\noindent
Consider the classical linear regression model,
\begin{equation} \label{y=Xbeta+epsilon}
\bm{y} = \bm{X} \bm{\beta} + \bm{\epsilon}, \numbereqn
\end{equation}
where $\bm{y}$ is an $n$-dimensional response vector, $\bm{X}_{n \times p} = [ \bm{X}_1, \ldots, \bm{X}_p ]$ is a fixed regression matrix with $n$ samples and $p$ covariates, $\bm{\beta} = (\beta_1, \ldots, \beta_p)'$ is a $p$-dimensional vector of unknown regression coefficients, and $\bm{\epsilon} \sim \mathcal{N}( \bm{0}, \sigma^2 \bm{I}_n)$, where $\sigma^2$ is the unknown variance. Throughout this paper, we assume that $\bm{y}$ and $\bm{X}$ have been centered at 0 so there is no intercept in our model.

In recent years, the high-dimensional setting when $p \gg n$ has received considerable attention. This scenario is now routinely encountered in areas as diverse as medicine, astronomy, and finance, just to name a few. In the Bayesian framework, there have been numerous methods proposed to handle the ``large $p$, small $n$'' scenario, including spike-and-slab priors with point masses at zero (e.g., \cite{MartinMessWalker2017, CastilloSchmidtHiebervanderVaart2015, YangWainwrightJordan2016}), continuous spike-and-slab priors (e.g., \cite{NarisettyHe2014, RockovaGeorge2018}), nonlocal priors (e.g. \cite{JohnsonRossell2012, RossellTelesca2017, ShinBhattacharyaJohnson2018}), and scale-mixture shrinkage priors (e.g. \cite{VanDerPasSalmondSchmidtHieber2016, SongLiang2017}). These priors have been shown to have excellent empirical performance and possess strong theoretical properties, including model selection consistency, (near) minimax posterior contraction, and Bernstein-von Mises theorems, among others. In this paper, we will restrict our focus to the scale-mixture shrinkage approach.

In the Bayesian literature, a popular method for estimating $\bm{\beta}$ in (\ref{y=Xbeta+epsilon}) when $p > n$ is to place scale-mixture shrinkage priors on the coefficients of interest and a prior on unknown variance, $\sigma^2$. These priors typically take the form,
\begin{equation} \label{scalemixture}
\begin{array}{c}
\beta_i | ( \sigma^2, \omega_i^2 ) \sim \mathcal{N}(0, \sigma^2 \omega_i^2), i = 1, \ldots, p, \\
\omega_i^2 \sim \pi(\omega_i^2), i = 1, \ldots, p, \\
\sigma^2 \sim \mu(\sigma^2),
\end{array}
\end{equation}
where $\pi: [0, \infty) \rightarrow [0, \infty)$ is a density on the positive reals. Priors of the form (\ref{scalemixture}) have been considered by many authors, e.g.,  \cite{ParkCasella2008, CarvalhoPolsonScott2010, GriffinBrown2010, BhattacharyaPatiPillaiDunson2015, BhadraDattaPolsonWillard2017, polsonscott2012, ArmaganClydeDunson2011, GriffinBrown2013, ArmaganDunsonLee2013}. 

From a computational perspective, scale-mixture priors are very attractive. Discontinuous spike-and-slab priors require searching over $2^p$ models, while continuous spike-and-slab and nonlocal priors result in multimodal posteriors. As a result, Markov chain Monte Carlo (MCMC) algorithms are prone to being trapped at a local posterior mode, and MCMC can suffer from slow convergence for these models. Scale-mixture shrinkage priors, on the other hand, do not face these drawbacks because they are continuous and typically give rise to unimodal posteriors. Additionally, there have been recent advances for fast MCMC sampling from normal scale-mixture priors that scale linearly in time with $p$, e.g. \cite{BhattacharyaChakrabortyMallick2016, JohndrowOrensteinBhattacharya2017}. 

Bayesian scale-mixture priors have been studied primarily under sparsity assumptions.  If sparse recovery of $\bm{\beta}$ is desired, the prior $\pi(\cdot)$ can be constructed so that it contains heavy mass around zero and heavy tails. This way, the posterior density $\pi(\bm{\beta} | \bm{Y})$ is also heavily concentrated around $\bm{0} \in \mathbb{R}^p$, while the heavy tails correctly identify and prevent overshrinkage of the true active covariates. 

While sparsity is often a reasonable assumption, it is not always appropriate, nor is there any ironclad reason to believe that sparsity is the true phenomenon. \citet{ZouHastie2005} demonstrated a practical example where the assumption of sparsity is violated: in microarray experiments with highly correlated predictors, it is often desirable for all genes which lie in the same biological pathway to be selected as a group, even if the final model is not parsimonious. \citet{ZouHastie2005} introduced the elastic net to overcome the inability of the LASSO \cite{Tibshirani1996} to select more than $n$ variables. In the Bayesian literature, there seems to be little study of the appropriateness of scale-mixture priors (\ref{scalemixture}) in dense settings. Ideally, we would like our priors on $\bm{\beta}$ in (\ref{y=Xbeta+epsilon}) to be able to handle \textit{both} sparse and non-sparse situations.

Another important issue to consider is the selection of hyperparameters in our priors on $\bm{\beta}$.  The empirical performance of Bayesian methods can be very sensitive to the choice of hyperparameters. Many authors, e.g. \cite{NarisettyHe2014, YangWainwrightJordan2016, MartinMessWalker2017}, have proposed fixing hyperparameters \textit{a priori} based on asymptotic arguments (such as consistency or minimaxity) or by minimizing some score function such as Bayesian information criterion (BIC) or deviance information criterion (DIC) (e.g. \cite{SongLiang2017, SpiegelhalterBestCarlinVanDerLinde2002}). In this paper, we will argue in favor of a different approach based on marginal maximum likelihood (MML) estimation. Our approach avoids the need for tuning by the user and allows our model to automatically adapt to the true underlying sparsity.

In this paper, we consider a scale mixture model (\ref{scalemixture}) with the beta prime density as the scale prior. We call our model the normal-beta prime (NBP) prior. \citet{BaiGhoshTesting2018} previously studied the NBP model in the context of multiple hypothesis testing of normal means.  Here, we extend the NBP prior to high-dimensional linear regression (\ref{y=Xbeta+epsilon}). Our main contributions are summarized as follows:
\begin{itemize}
\item
We show that for high-dimensional linear regression, the NBP model can serve as both a sparse \textit{and} a non-sparse prior. We prove that under sparsity and appropriate regularity conditions, the NBP prior asymptotically obtains the (near) minimax posterior contraction rate. 
\item
In the absence of prior knowledge about sparsity or non-sparsity, we propose an empirical Bayes variant of the NBP model which enables our model to be \textit{self-adaptive} and learn the true sparsity level from the data. Under our procedure, the hyperparameter estimates are never at risk of collapsing to zero. This is not the case for many other choices of priors, where empirical Bayes estimates can often result in estimates of zero and thus degenerate priors. 
\item
We derive efficient Monte Carlo EM and variational EM algorithms for implementing the self-adaptive NBP model. Our algorithms embed the EM algorithm for estimating the hyperparameters into posterior simulation updates, so that the hyperparameters do not need to be tuned separately.
\end{itemize}

\subsection{Notation}
We use the following notations for the rest of the paper. Let $\{ a_n \}$ and $\{ b_n \}$ be two  non-negative sequences of real numbers indexed by $n$, where $b_n \neq 0$ for sufficiently large $n$. We write $a_n \asymp b_n$ to denote $0 < \lim \inf_{n \rightarrow \infty} a_n/b_n \leq \lim \sup_{n \rightarrow \infty} a_n/b_n < \infty$. If $\lim_{n \rightarrow \infty} a_n/b_n = 0$, we write $a_n = o(b_n)$ or $a_n \prec b_n$. We use $a_n \lesssim b_n$ or $a_n = O(b_n)$ to denote that for sufficiently large $n$, there exists a constant $C >0$ independent of $n$ such that $a_n \leq Cb_n$ respectively. 

For a vector $\bm{v} \in \mathbb{R}^p$, we let $|| \bm{v} ||_0 := \sum_{i=1}^p \bm{1}(v_i \neq 0)$, $|| \bm{v} ||_1 := \sum_{i=1}^p |v_i|$, and $|| \bm{v} ||_2 := \sqrt{ \sum_{i=1}^p v_i^2}$ denote the $\ell_0$, $\ell_1$, and $\ell_2$ norms respectively. For a set $\mathcal{A}$, we denotes its cardinality as $| \mathcal{A} |$.

\section{The Normal-Beta Prime (NBP) Model} \label{NBPPrior}
The beta prime density is given by
\begin{equation} \label{betaprime}
\pi(\omega_i^2) = \frac{\Gamma(a+b)}{\Gamma(a) \Gamma(b)} (\omega_i^2)^{a-1} ( 1 + \omega_i^2)^{-a-b}. 
\end{equation}
 In particular, setting $a=b=0.5$ in (\ref{betaprime}) yields the half-Cauchy prior $\mathcal{C}^+ (0,1)$ for $\omega_i$. In the normal means setting with $\bm{X} = \bm{I}$, $p = n$, and $\sigma^2 = 1$ in (\ref{y=Xbeta+epsilon}), \citet{polsonscott2012} conducted numerical experiments for different combinations of $(a, b)$ in (\ref{betaprime}) and argued that the half-Cauchy prior should be a default prior for the top-level scale parameter in Bayesian hierarchical models.  \citet{PerezPericchiRamirez2017} also generalized the beta prime density (\ref{betaprime}) to the Scaled Beta2 (SBeta2) family of scale priors by adding an additional scaling parameter to (\ref{betaprime}). \citet{polsonscott2012} and  \citet{PerezPericchiRamirez2017} did not consider linear regression models under general design matrices. 

The normal-beta prime (NBP) model is as follows. Suppose we place a normal-scale mixture prior with the beta prime density (\ref{betaprime}) as the scale parameter for each of the individual coefficients in $\bm{\beta} = (\beta_1, \ldots, \beta_p)$ and the usual inverse gamma prior $\mathcal{IG}(c, d)$ prior on $\sigma^2$, where $c, d > 0$. Letting $\beta ' (a, b)$ denote the beta prime distribution (\ref{betaprime}) with hyperparameters $a > 0$, $b > 0$, our Bayesian hierarchy is as follows:
\begin{equation} \label{NBPhier}
\begin{array}{cl}
\beta_i | \omega_i^2, \sigma^2 \sim \mathcal{N}(0, \sigma^2 \omega_i^2), & i = 1, \ldots, p, \\
\omega_i^2 \sim \beta ' (a, b), & i = 1, \ldots, p, \\
\sigma^2 \sim \mathcal{IG}(c, d), &
\end{array}
\end{equation}
For our model (\ref{NBPhier}), we can choose very small values of $c$ and $d$ in order to make the prior on $\sigma^2$ relatively noninfluential and noninformative (e.g., a good default choice is $c=d=10^{-5}$). The most critical hyperparameter choices governing the performance of our model are $(a, b)$.
\begin{proposition}
\label{Prop:2.1}
Suppose that we endow $(\bm{\beta}, \sigma^2)$ with the priors in (\ref{NBPhier}). Then the marginal distribution, $\pi(\beta_i | \sigma^2), i = 1, \ldots p,$  is unbounded with a singularity at zero for any $0 < a \leq 1/2$. 
\end{proposition} 
\begin{proof} 
See Proposition 2.1 in \citet{BaiGhoshTesting2018}.
\end{proof}
Proposition \ref{Prop:2.1} implies that in order to facilitate sparse recovery of $\bm{\beta}$, we should set the hyperparameter $a$ to be a small value. This would force the NBP prior to place most of its mass near zero, and thus, the posterior $\pi(\bm{\beta} | \bm{y})$ would also be concentrated near $\bm{0} \in \mathbb{R}^p$. 

\begin{figure}[t!]
\centering
\includegraphics[scale=0.6]{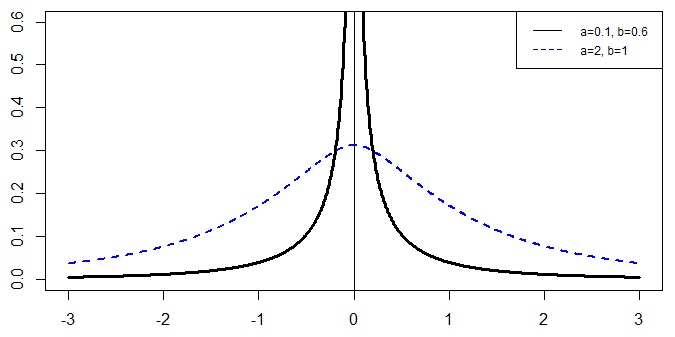} 
\caption{The marginal densities of the NBP prior, $\pi( \beta | \sigma^2)$, with $\sigma^2 = 1$. A smaller $b$ leads to a pole at zero and most of the mass near zero, whereas a large $b$ removes the singularity and leads to heavier tails. }
\label{fig:1}
\end{figure}

Figure \ref{fig:1} plots the marginal density, $\pi ( \beta | \sigma^2)$, for a single $\beta$, where we set $\sigma^2 = 1$ for illustration. When $a = 0.1$, the marginal density contains a singularity at zero, and the probability mass is heavily concentrated near zero. However, when $a = 2$, the marginal density does not contain a pole at zero, and the tails are significantly heavier. 

Figure \ref{fig:1} shows that the NBP model can serve as both a sparse and a non-sparse prior. If we have a priori knowledge that the true model is sparse with a few large signal values, we can fix $a$ to be a very small value. On the other hand, if the true model is known to be dense, we can set $a$ to a larger value, so we have a more diffuse prior. Then there would be less shrinkage of individual covariates in the posterior distribution. In Section \ref{mmle}, we will introduce the self-adaptive NBP model, which automatically learns the true sparsity level from the data and avoids the need for tuning by the user. 

\section{Posterior Contraction Rates Under the NBP Prior} \label{contractionrates}

\subsection{Preliminaries}

For our theoretical analysis of the NBP prior, we shall be principally concerned with the case when $p$ diverges to infinity as sample size $n$ grows and the underlying model is sparse. For the remainder of this section, we rewrite $p$ as $p_n$ to emphasize its dependence on $n$. We work under the frequentist assumption that there is a true data-generating model, i.e.,
\begin{equation} \label{truemodel}
\bm{y}_n = \bm{X}_n \bm{\beta}_{0} + \bm{\epsilon}_n,
\end{equation}
where $\bm{\epsilon}_n \sim \mathcal{N}(\bm{0}, \sigma_0^2 \bm{I}_n)$ and $\sigma_0^2$ is a fixed noise parameter.

Let $s_n = || \bm{\beta}_0 ||_0$ denote the size of the true model, and suppose that $s_n = o( n / \log p_n)$. Under (\ref{truemodel}) and appropriate regularity conditions, \citet{RaskuttiWainwrightYu2011} showed that the minimax estimation rate for any point estimator $\widehat{\bm{\beta}}$ of $\bm{\beta}_0$ under $\ell_2$ error loss is $\sqrt{s_n \log (p_n / s_n ) /n }$. Many frequentist point estimators such as the LASSO \cite{Tibshirani1996} estimator have been shown to attain the \textit{near}-minimax rate of $\sqrt{s_n \log p_n /n }$ under $\ell_2$ error loss.

In the Bayesian paradigm, on the other hand, we are mainly concerned with the rate at which the \textit{entire} posterior distribution contracts around the true $\bm{\beta}_0$. Letting $\mathbb{P}_0$ denote the probability measure underlying (\ref{truemodel}) and $\Pi( \bm{\beta} | \bm{y}_n)$ denote the posterior distribution of $\bm{\beta}$, our aim is to find a positive sequence $r_n$ such that
\begin{equation*}
\Pi ( \bm{\beta}: || \bm{\beta} - \bm{\beta}_0 || \geq M r_n | \bm{y}_n ) \rightarrow 0 \textrm{ a.s. } \mathbb{P}_0 \textrm{ as } n \rightarrow \infty, 
\end{equation*}
for some constant $M > 0 $. The frequentist minimax convergence rate is a useful benchmark for the speed of contraction $r_n$, since the posterior cannot contract faster than the minimax rate \cite{GhosalGhoshVanDerVaart2000}. 

Additionally, we are interested in posterior \textit{compressibility} \cite{BhattacharyaPatiPillaiDunson2015}, which allows us to quantify how well the NBP posterior captures the true sparsity level $s_n$. Since the NBP prior is absolutely continuous, it assigns zero mass to exactly sparse vectors. To approximate the model size for the NBP model, we use the following generalized notion of sparsity \cite{Rockova2018, RockovaGeorge2018, BhattacharyaPatiPillaiDunson2015}. Letting $\delta$ be some positive constant (to be specified later), we define the generalized inclusion indicator and generalized dimensionality, respectively, as
\begin{equation} \label{generalizeddimensionality}
\gamma_{\delta} (\beta) = I( | \beta / \sigma | > \delta) \textrm{ and } | \bm{\gamma}_{\delta} ( \bm{\beta} ) | = \displaystyle \sum_{i=1}^{p_n} \gamma_{\delta} (\beta_i).
\end{equation}
The generalized dimensionality counts the number of covariates in $\bm{\beta}/ \sigma$ that fall outside the interval $[-\delta, +\delta]$. With appropriate choice of $\delta$, the prior is said to have the posterior compressibility property if the probability that $| \bm{\gamma}_{\delta} (\bm{\beta}) |$ asymptotically exceeds a constant multiple of the true sparsity level $s_n$ tends to 0 as $n \rightarrow \infty$, i.e.
\begin{equation*} 
\Pi( \bm{\beta}: | \bm{\gamma}_{\delta} ( \bm{\beta} ) | \geq A s_n | \bm{y}_n ) \rightarrow 0 \textrm{ a.s. } \mathbb{P}_0 \textrm{ as } n \rightarrow \infty,
\end{equation*}
for some constant $ A> 0$.

\subsubsection{Near-Minimax Posterior Contraction Under the NBP Prior} \label{minimaxregcontraction}

We first introduce the following set of regularity conditions, which come from \citet{SongLiang2017} and which are fairly standard in the high-dimensional literature. As before, $s_n$ denotes the size of the true model, while $\lambda_{\min} (\bm{A})$ denotes the minimum eigenvalue of a symmetric matrix $\bm{A}$.

\subsubsection*{Regularity conditions }
\begin{enumerate}[label=(A\arabic*)]
\item
All the covariates are uniformly bounded. For simplicity, we assume they are all bounded by 1. \label{As:A1}
\vspace{-.2cm}
\item
$p_n \gg n$. \label{As:A2}
\vspace{-.2cm}
\item
Let $\xi \subset \{ 1, \ldots, p_n \}$, and let $\bm{X}_{\xi}$ denote the submatrix of $\bm{X}_n$ that contains the columns with indices in $\xi$. There exists some integer $\bar{p}$ (depending on $n$ and $p_n)$ and fixed constant $t_0$ such that $\bar{p} \succ s$ and $\lambda_{\min} ( \bm{X}_{\xi}^\top \bm{X}_{\xi} ) \geq n t_0$ for any model of size $| \xi | \leq \bar{p}$. \label{As:A3} 
\vspace{-.2cm}
\item
$s_n = o(n / \log p_n)$. \label{As:A4} 
\vspace{-.2cm}
\item
$\max_j \{ | \beta_{0j} / \sigma_0 | \} \leq \gamma_3 E_n$ for some $\gamma_3 \in (0, 1)$, and $E_n$ is nondecreasing with respect to $n$.  \label{As:A5}
\end{enumerate}

\noindent Assumption \ref{As:A3} is a minimum restricted eigenvalue (RE) condition which ensures that $\bm{X}_n^\top \bm{X}_n$ is locally invertible over sparse sets. When $p_n \gg n$, minimum RE conditions are imposed to render $\bm{\beta}_0$ estimable. Assumption \ref{As:A4} restricts the growth of $s_n$, and \ref{As:A5} constrains the size of the signals in $\bm{\beta}_0$ to be $O(E_n)$ for some nondecreasing sequence $E_n$.

As we illustrated in Section \ref{NBPPrior}, the hyperparameter $a$ is mainly what controls the amount of posterior mass around 0 for each coefficient $\beta_i, i = 1, \ldots, p_n$, under the NBP prior. Hence, it will play a crucial role in our theory. We rewrite $a$ as $a_n$ to emphasize its dependence on $n$.

\begin{theorem}
\label{Th:3.1}
Assume that Assumptions \ref{As:A1}-\ref{As:A5} hold, with $\log (E_n) = O(\log p_n)$ for Assumption \ref{As:A5}. Let $r_n = M \sqrt{s_n \log p_n / n}$ for some fixed constant $M > 0$, and let $k_n \asymp (\sqrt{s_n \log p_n / n }) / p_n$. Suppose that we place the NBP prior (\ref{NBPhier}) on $(\bm{\beta}, \sigma^2)$, with $a_n \lesssim k_n^2 p_n^{-(1+u)}$, for some $u > 0$, and $b \in (1, \infty)$. Then under (\ref{truemodel}), the following hold:
\begin{equation} \label{L2loss}
\Pi \left( \bm{\beta}: || \bm{\beta} - \bm{\beta}_0 ||_2 \geq c_1  \sigma_0 r_n | \bm{y}_n \right) \rightarrow 0 \textrm{ a.s. } \mathbb{P}_0 \textrm{ as } n\rightarrow \infty,
\end{equation}
\begin{equation} \label{L1loss}
\Pi \left( \bm{\beta}: || \bm{\beta} - \bm{\beta}_0 ||_1 \geq c_1 \sigma_0 \sqrt{s} r_n | \bm{y}_n \right) \rightarrow 0 \textrm{ a.s. } \mathbb{P}_0 \textrm{ as } n \rightarrow \infty,
\end{equation}
\begin{equation} \label{predictionloss}
\Pi \left( \bm{\beta}: || \bm{X} \bm{\beta} - \bm{X} \bm{\beta}_0 ||_2 \geq c_0 \sigma_0 \sqrt{n} r_n | \bm{Y}_n \right) \rightarrow 0 \textrm{ a.s. } \mathbb{P}_0 \textrm{ as } n \rightarrow \infty,
\end{equation}
\begin{equation} \label{compressibility}
\Pi \left( \bm{\beta}: | \bm{\gamma}_{k_n} ( \bm{\beta} ) | \geq \widetilde{q}_n | \bm{y}_n \right) \rightarrow 0 \textrm{ a.s. } \mathbb{P}_0 \textrm{ as } n \rightarrow \infty,
\end{equation}
where $c_0 > 0, c_1 > 0$, $| \gamma_{k_n}(\bm{\beta}) | = \sum_{i=1}^{p} I( | \beta_i / \sigma | > k_n)$, and $\widetilde{q}_n \asymp s_n$.
\end{theorem}
\begin{proof}
See Supplementary Materials, Appendix \ref{appA}.
\end{proof}

The proof of Theorem \ref{Th:3.1} is based on verifying a set of conditions by \citet{SongLiang2017}. In particular, (\ref{L2loss})-(\ref{predictionloss}) show that by fixing $a_n \lesssim p_n^{-(3+u)} \sqrt{s_n \log p_n / n } $, for some $u > 0$, and $b \in (1, \infty)$ as the hyperparameters $(a_n, b)$ in (\ref{NBPhier}), the NBP model's posterior contraction rates under $\ell_2$, $\ell_1$, and prediction error loss are the familiar near-optimal rates of $\sqrt{s_n \log p_n /n }$, $s_n \sqrt{ \log p_n / n }$, and $\sqrt{ s_n \log p_n }$ respectively. Moreover, by setting $\delta = k_n \asymp (\sqrt{s_n \log p_n / n }) / p_n$ in our generalized inclusion indicator (\ref{generalizeddimensionality}), (\ref{compressibility}) also shows that the NBP possesses posterior compressibility, i.e. the probability that the generalized dimension size $|\bm{\gamma}_{k_n} (\bm{\beta}) |$ is a constant multiple larger than $s_n$ asymptotically vanishes. Note that $k_n \approx 0$ for large $n$, so that $ | \bm{\gamma}_{k_n} (\bm{\beta}) |$ is a very good approximation of the limiting ideal of $ || \bm{\beta} / \sigma ||_0$. 

Our result relies on setting the hyperparameter $a_n$ to be a value dependent upon the unknown sparsity level $s_n$. Previous theoretical results for scale-mixture shrinkage priors, e.g. \cite{VanDerPasSalmondSchmidtHieber2016, SongLiang2017}, also rely on fixing hyperparameters to quantities that depend on $s_n$ in order for these priors to obtain minimax posterior contraction. If we want to \textit{a priori} fix the hyperparameters $(a, b)$ under the NBP prior based on asymptotic arguments, we could first obtain an estimate of $s_n$. For example, we could take $\widehat{s}_n = || \widehat{\bm{\beta}}^{ALasso} ||_0$, where $\widehat{\bm{\beta}}^{ALasso}$ is an adaptive LASSO solution \cite{Zou2006} to (\ref{y=Xbeta+epsilon}) and then set $a_n$ as $a_n := p_n^{-(3+u)} \sqrt{ \widehat{s}_n \log p_n / n }, u>0$. Fixing $a_n := p_n^{-(3+u)} \sqrt{ \log n / n}, u>0$, would also satisfy the conditions in our theorem (since $\log n \prec s_n \log p_n$), thus removing the need to estimate $s_n$.

\section{Empirical Bayes Estimation of Hyperparameters} \label{mmle}

While fixing $(a, b)$ \textit{a priori} as $a = p^{-(3+u)} \sqrt{\log n / n }$, for some $u>0$, and $b \in (1, \infty)$ would lead to (near) minimax posterior contraction under conditions \ref{As:A1}-\ref{As:A5}, this would not allow the NBP prior to adapt to varying patterns of sparsity or signal strengths. The minimum restricted eigenvalue assumption \ref{As:A3} is also computationally infeasible to verify in practice. \citet{DobribanFan2016} showed that, given an arbitrary design matrix $\bm{X}$, verifying that the minimum RE condition holds is an NP-hard problem. Finally, there is no practical way of verifying that the model size condition \ref{As:A4} that $s  = o(n/ \log p)$ holds, or that the true model is even sparse.

For these reasons, we do not recommend fixing the hyperparameters $(a, b)$ in the NBP model based on asymptotic arguments. We instead prefer to \textit{learn} the true sparsity pattern from the given data. One such way to do this is to use marginal maximum likelihood (MML). The marginal likelihood, $f(\bm{y}) = \int f(\bm{y} | \bm{\beta}, \sigma^2) \pi(\bm{\beta}, \sigma^2) d ( \bm{\beta}, \sigma^2 )$, is the probability the model gives to the observed data with respect to the prior (or the ``model evidence''). Hence, choosing the prior hyperparameters $(a, b)$ to maximize $f(\bm{y})$ gives the maximum ``model evidence,'' and the MML can learn the most likely sparsity level (or non-sparsity) from the data. One potential shortcoming with MML is that  it can lead to degenerate priors. However, as we illustrate below, this problem is avoided under the NBP prior.

We propose an EM algorithm to obtain the MML estimates of $(a, b)$. Henceforth, we refer to this empirical Bayes variant of the NBP model as the \textit{self-adaptive} NBP model. Our algorithm can be easily incorporated into Gibbs sampling or mean field variational Bayes (MFVB) algorithms. 

To construct the EM algorithm, first note that because the beta prime density can be rewritten as a product of an independent gamma and inverse gamma densities, we may reparametrize (\ref{NBPhier}) as

\begin{equation} \label{NBPhierreparem}
\begin{array}{cl}
\beta_i | ( \omega_i^2, \lambda_i^2 \xi_i^2 ) \sim \mathcal{N}(0, \sigma^2 \lambda_i^2 \xi_i^2), & i = 1, \ldots, p, \\
\lambda_i^2 \sim \mathcal{G} (a, 1), & i = 1, \ldots, p, \\
\xi_i^2 \sim \mathcal{IG} (b, 1), & i = 1, \ldots, p, \\
\sigma^2  \sim \mathcal{IG}(c, d). &
\end{array}
\end{equation}
The logarithm of the joint posterior under the reparametrized NBP prior (\ref{NBPhierreparem}) is given by 
\begin{align*} \label{logposteriorNBP}
& - \left( \frac{n+p}{2} \right) \log ( 2 \pi) - \left( \frac{n+p}{2} + c+ 1 \right) \log (\sigma^2) - \frac{1}{2 \sigma^2} || \bm{y} - \bm{X} \bm{\beta} ||_2^2  \\
& \qquad - \displaystyle \sum_{i=1}^p \frac{\beta_i^2}{2 \lambda_i^2  \xi_i^2 \sigma^2 } - p \log ( \Gamma(a)) + \left( a - \frac{3}{2} \right) \displaystyle \sum_{i=1}^p \log (\lambda_i^2) - \displaystyle \sum_{i=1}^p \lambda_i^2    \\
& \qquad - p \log (\Gamma(b)) - \left( b + \frac{3}{2} \right) \displaystyle \sum_{i=1}^p \log ( \xi_i^2 ) - \displaystyle \sum_{i=1}^{p} \frac{1}{\xi_i^2} + c \log (d) - \log (\Gamma(c)) - \frac{d}{\sigma^2}.  \numbereqn
\end{align*} 
Thus, at the $k$th iteration the EM algorithm, the conditional log-likelihood on $\nu^{(k-1)} = (a^{(k-1)}, b^{(k-1)})$ and $\bm{y}$ in the E-step is given by
\begin{align*} \label{Estep}
 Q( \nu | \nu^{(k-1)} ) & = -p \log( \Gamma(a)) + a \displaystyle \sum_{i=1}^p \mathbb{E}_{a^{(k-1)}} \left[ \log (\lambda_i^2) | \bm{y} \right] - p (\log \Gamma(b)) \\ 
& - b \displaystyle \sum_{i=1}^p  \mathbb{E}_{b^{(k-1)}} \left[ \log (\xi_i^2) | \bm{y} \right]  + \textrm{ terms not involving } a \textrm{ or } b. \numbereqn
\end{align*}
The M-step maximizes $Q(\nu | \nu^{(k-1)})$ over $\nu = (a, b)$ to produce the next estimate $\nu^{(k)} = (a^{(k)}, b^{(k)})$. That is, we must find $(a, b)$, $a \geq 0, b \geq 0$, such that
\begin{equation} \label{Mstep}
\begin{array}{rl}
\frac{ \partial Q}{ \partial a} & = -p \psi(a) + \displaystyle \sum_{i=1}^p \mathbb{E}_{a^{(k-1)}} \left[ \log (\lambda_i^2) | \bm{y} \right]  =  0, \\
\frac{ \partial Q}{ \partial b}  & = -p \psi(b) - \displaystyle \sum_{i=1}^p \mathbb{E}_{b^{(k-1)}} \left[ \log (\xi_i^2) | \bm{y} \right] = 0, 
\end{array}
\end{equation}
where $\psi(x) = d/dx \left( \Gamma(x) \right)$ denotes the digamma function. We can solve for $(a,b)$ in (\ref{Mstep}) numerically by using a fast root-finding algorithm such as Newton's method. The summands, $\mathbb{E}_{a^{(k-1)}} \left[ \log ( \lambda_i^2 ) | \bm{y} \right]$ and $\mathbb{E}_{b^{(k-1)}} \left[ \log ( \xi_i^2 ) | \bm{y} \right]$, $i = 1, \ldots, p,$ in (\ref{Mstep}) can be estimated from either the mean of $M$ Gibbs samples based on $\nu^{(k-1)}$, for sufficiently large $M > 0$ (as in \cite{Casella2001}), or from the $(k-1)$st iteration of the mean field variational Bayes (MFVB) algorithm (as in \cite{Leday2017}).

\begin{theorem}
\label{Th:4.1}
At every $k$th iteration of the EM algorithm for the self-adaptive NBP model, there exists a unique solution $\nu^{(k)} = (a^{(k)}, b^{(k)})$,  which maximizes (\ref{Estep}) in the M-step. Moreover, $a^{(k)} > 0$, $b^{(k)} > 0$ at the $k$th iteration.
\end{theorem}
\begin{proof}
See Supplementary Materials, Appendix \ref{appA}.
\end{proof}
Theorem \ref{Th:4.1} ensures that under our setup, we will not encounter the issue of the sparsity parameter $a$ (or the parameter $b$) collapsing to zero. Empirical Bayes estimates of zero are a major concern for MML approaches to estimating hyperparameters in Bayesian linear regression models. For example, in $g$-priors,
\begin{align*}
\bm{\beta} | \sigma^2 \sim \mathcal{N}_p \left( \bm{\gamma},  g \sigma^2  ( \bm{X}^\top \bm{X})^{-1} \right),
\end{align*}
\citet{GeorgeFoster2000} showed that the MML estimate of the parameter $g$ could equal zero. In global-local shrinkage priors of the form,
\begin{align*}
\beta_i | (\lambda_i^2, \sigma^2) \sim \mathcal{N}(0, \sigma^2 \tau^2 \lambda_i^2), \hspace{.2cm} \lambda_i^2 \sim \pi(\lambda_i^2), \hspace{.2cm} i = 1, \ldots, p,
\end{align*}
the variance rescaling parameter $\tau$ is also at risk of being estimated as zero under MML \cite{PolsonScott2010, TiaoTan1965, CarvalhoPolsonScott2009, DattaGhosh2013}. Finally, \citet{ScottBerger2010} proved that if we endow (\ref{y=Xbeta+epsilon}) with a binomial model selection prior,
\begin{align*}
\pi ( \bm{M}_{\gamma} | \theta ) = \theta^{k_{\gamma}} (1-\theta)^{p-k_{\gamma}},
\end{align*}
where $\bm{M}_{\gamma}$ is the model indexed by $\gamma \subset \{1, \ldots, p \}$ and  $k_{\gamma}$ represents the number of included variables in the model, the MML estimate of the mixing proportion $\theta$ could be estimated as either 0 or 1, leading to a degenerate prior. Clearly, the marginal maximum likelihood approach for tuning hyperparameters is not without problems, as it could potentially lead to degenerate priors in high-dimensional regression. However, with the NBP prior, we can easily incorporate a data-adaptive procedure for estimating the hyperparameters while avoiding this potential pitfall. 

In the aforementioned examples, placing priors on $g$, $\tau$, or $\theta$ with strictly positive support or performing cross-validation or \textit{restricted} marginal maximum likelihood estimation over a range of strictly positive values can help to avoid the issue of collapse to zero. The hierarchical Bayes approach does not quite address the issue of misspecification of hyperparameters (since these still need to be specified in the additional priors).  If we use cross-validation over a grid of positive values, the ``optimal'' choice or spacing of grid points is also not clear-cut.

In the general regression setting, it is also unclear what the endpoints should be if we use a truncated range of positive values to estimate hyperparameters from restricted marginal maximum likelihood. Recently, for sparse normal means estimation (i.e. $\bm{X} = \bm{I}$, $p = n$, and $\sigma^2 = 1$ in (\ref{y=Xbeta+epsilon})), \citet{VanDerPasSzaboVanDerVaart2017} advocated using the restricted MML estimator for the sparsity parameter $\tau$ in the range $[1/n, 1]$ for the horseshoe prior \cite{CarvalhoPolsonScott2010}. This choice allows the horseshoe model to obtain the (near) minimax posterior contraction rate. While this choice gives theoretical guarantees for normal means estimation, it does not seem to be justified for high-dimensional regression (\ref{y=Xbeta+epsilon}) when $p \gg n$. Theorem 3.1 in \citet{SongLiang2017} shows that the minimax optimal choice for $\tau$ in the horseshoe under model (\ref{y=Xbeta+epsilon}) satisfies $\tau \lesssim (\sqrt{s_n \log p / n})  p^{-(1+(u+1)/(r-1))}$, where $u > 0$, $r > 0$. It would thus appear that any $\tau \in [1/n, 1]$ would lead to a \textit{suboptimal} contraction rate for sparse high-dimensional regression. In our numerical experiments in Section \ref{Simulations}, we demonstrate that the truncation suggested by \citet{VanDerPasSzaboVanDerVaart2017} leads to worse estimation for the horseshoe under the general linear regression model (\ref{y=Xbeta+epsilon}) (as opposed to normal means estimation).
 
The self-adaptive NBP prior circumvents all of these issues by obtaining the MML estimates of $(a, b)$ over the entire range $[0, \infty) \times [0, \infty)$, while ensuring that $(a, b)$ are never estimated as zero. Thus, the self-adaptive NBP model's automatic selection of hyperparameters through MML provides a practical alternative to hierarchical Bayes, cross-validation, or restricted MML approaches for tuning hyperparameters. 

\subsection{Illustration of the Self-Adaptive NBP Model} \label{selfadaptivity}

We now illustrate the self-adaptive NBP prior's ability to adapt to differing sparsity patterns. We consider two settings: one sparse ($n=60$, $p=100$, 10 nonzero covariates) and one dense ($n = 60$, $p=100$, and 60 nonzero covariates), where the active covariates are drawn from $\mathcal{U}\left( [-2, -0.5] \cup [0.5, 2] \right)$. Our simulations come from experiments 1 and 4 in Section \ref{Simulations}. We initialize $(a^{(0)}, b^{(0)}) = (0.01, 0.01)$ and then implement the Monte Carlo EM algorithm (described in Section \ref{PosteriorSimulation}) for finding the MML estimates of the parameters $(a, b)$, which we denote as $(\widehat{a}, \widehat{b})$. 

In Figure \ref{fig:2}, we plot the iterations from two runs of the EM algorithm.  The algorithm terminates at iteration $k$ when the square of the $\ell_2$ distance between $(a^{(k-1)}, b^{(k-1)})$ and $(a^{(k)}, b^{(k)})$ reaches below $10^{-6}$. We then set $(\widehat{a}, \widehat{b}) = (a^{(k)}, b^{(k)})$. The top panel in Figure \ref{fig:2} plots the paths for $a$ and $b$ from a sparse model with 10 active predictors, and the bottom panel plots the paths for $a$ and $b$ from a dense model with 60 active predictors. In the sparse case, the final MML estimate of $a$ is $\widehat{a} = 0.184$. In the dense case, the final MML estimate of $a$ is $\widehat{a} = 1.104$. 

Figure \ref{fig:3} shows the NBP's marginal density, $\pi( \beta | \widehat{a}, \widehat{b}, \sigma^2)$, for a single coefficient $\beta$ using the MML estimates of $(a, b)$ obtained in sparse and the dense settings respectively. For the purpose of illustration, we have fixed $\sigma^2 = 1$. The left panel depicts the marginal density under the sparse setting (10 active predictors, $(\widehat{a}, \widehat{b}) = (0.184, 1.124)$). We see that the marginal density for $\beta$ contains a singularity at zero, \textit{and} most of the probability mass is around zero. We thus recover a sparse model for $\pi(\bm{\beta} | \bm{y})$ under these MML hyperparameters. Meanwhile, the right panel depicts the marginal density in the dense setting (60 active predictors, $(\widehat{a}, \widehat{b}) = (1.104, 1.645)$). Here, the marginal density for $\beta$ does \textit{not} contain a pole, and more mass is placed in neighborhoods away from zero. Thus, we recover a more dense model. Figures \ref{fig:2} and \ref{fig:3} illustrate that in both cases, the EM algorithm was able to correctly learn the true sparsity (or non-sparsity) from the data and incorporate this into its estimates of the hyperparameters.

\begin{figure}[t!]
	\centering
	\includegraphics[scale=0.5, trim={0 0.5cm 0 0.25cm}]{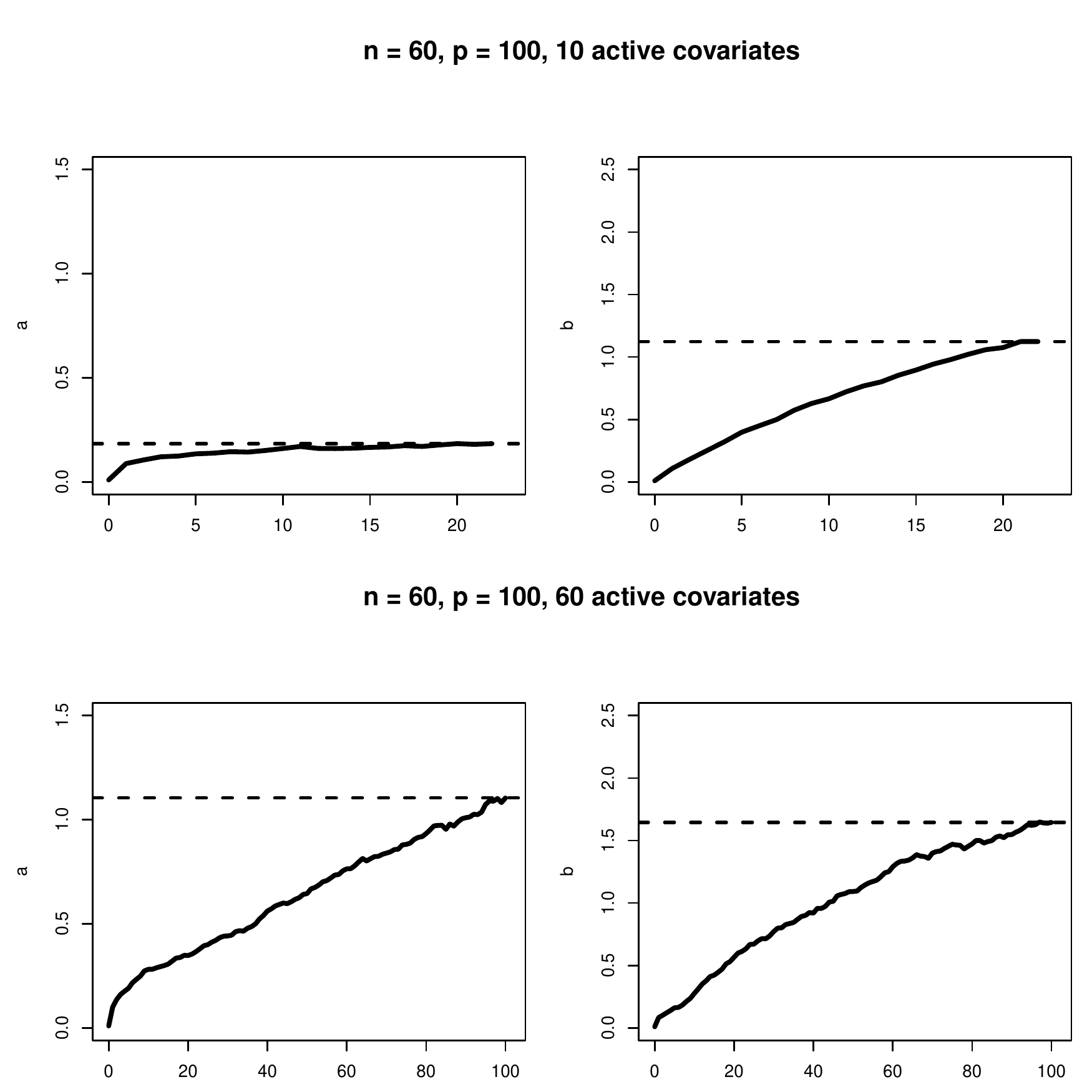} 
	\caption{Paths of the Monte Carlo/EM algorithm for estimating the MML estimators of $a$ and $b$ for one sparse case and one dense case.  The dashed line indicates the final MML estimate at convergence.  }
	\label{fig:2}
\end{figure}

\begin{figure}[h!]
	\centering
	\includegraphics[scale=0.65]{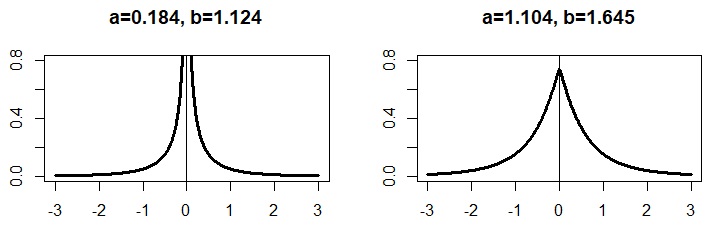}
	\caption{The marginal densities of the NBP prior, $\pi( \beta | a, b, \sigma^2)$, with different MML estimates of $(a, b)$. }
	\label{fig:3}
\end{figure}

A referee has pointed out that placing a mixture prior of beta prime densities as the prior for $\omega_i$'s in (\ref{scalemixture}) could also accommodate dense situations. While we recognize this fact, we believe that it is better to use marginal maximum likelihood (MML). First, putting a mixture of beta primes as the prior on $\omega_i^2, i=1,\ldots, p,$ would make the posteriors for $\beta_i, i = 1, \ldots, p$, multimodal.  The quality of our posterior approximation algorithms in Section \ref{PosteriorComputation} is dependent on the assumption that the approximate posterior is unimodal  (especially if we use a variational density to approximate $\pi(\bm{\beta} | \bm{y})$). Second, if we used a mixture prior, we would then need to tune both the mixture weight(s) and the hyperparameters in each mixture component. As we demonstrate in Sections \ref{selfadaptivity} and \ref{Simulations}, utilizing a single beta prime prior as the scale with MML estimates for hyperparameters performs quite well.

\section{Computation for the NBP Model} \label{PosteriorComputation}

\subsection{Posterior Simulation} \label{PosteriorSimulation}
Using the reparametrization (\ref{NBPhierreparem}), we see that the NBP model admits the following full conditional densities. Let $\bm{D} = \textrm{diag} ( \lambda_1^2 \xi_1^2, \ldots, \lambda_p^2 \xi_p^2 )$. The full conditional densities under the NBP model are:
\begin{equation} \label{NBPconditionals}
\begin{array}{rl}
\bm{\beta} \rvert \textrm{ rest} & \sim \mathcal{N}_p \left( \left( \bm{X}^\top \bm{X} + \bm{D}^{-1} \right)^{-1} \bm{X}^\top \bm{y}, \sigma^2 \left( \bm{X}^\top \bm{X} + \bm{D}^{-1} \right)^{-1} \right), \\
\lambda_i^2 \rvert \textrm{ rest} & \overset{ind}{\sim} \mathcal{GIG} \left( \frac{\beta_i^2}{\sigma^2 \xi_i^2}, 2, a - \frac{1}{2} \right),\hspace{.5cm}  i = 1, \ldots, p, \\
\xi_i^2 \rvert \textrm{ rest} & \overset{ind}{\sim} \mathcal{IG} \left( b + \frac{1}{2}, \frac{\beta_i^2}{2 \sigma^2 \lambda_i^2} + 1  \right), \hspace{.5cm} i = 1, \ldots, p, \\
\sigma^2 \rvert \textrm{ rest} & \sim \mathcal{IG} \left( \frac{n+p+2c}{2}, \frac{ || \bm{y} - \bm{X} \bm{\beta} ||_2^2 + \bm{\beta}^\top \bm{D}^{-1} \bm{\beta} + 2d}{2} \right), 
\end{array}
\end{equation}
where $\mathcal{GIG}(a, b, p)$ denotes a generalized inverse Gaussian density with the pdf, $f(x; u, v, p) \propto x^{(p-1)} e^{-(u/x+vx)/2}$. From (\ref{NBPconditionals}), implementation through Gibbs sampling is straightforward. Moreover, to save on computational time, the $\lambda_i$'s and $\xi_i$'s, $i = 1, \ldots, p$, are block-updated in parallel, and we can utilize the fast sampling algorithm of \citet{BhattacharyaChakrabortyMallick2016} to sample from the full conditional for $\bm{\beta}$ in $O(n^2 p)$ time.

To incorporate the EM algorithm for estimating $(a, b)$ from Section \ref{mmle} into our Gibbs sampler, we update $(a, b)$ every $M = 100$ iterations of the Gibbs sampler by solving (\ref{Mstep}) and estimating the summand terms in (\ref{Mstep}) from the mean of the past $M$ iterations of the Gibbs sample. Complete technical details for our Monte Carlo EM algorithm are given in Appendix \ref{AppB} of the Supplementary Materials.

The conditionals (\ref{NBPconditionals}) also admit a mean field variational Bayes (MFVB) implementation for the NBP model. Let $\bm{\lambda} = ( \lambda_1^2, \ldots, \lambda_p^2 )$ and $\bm{\xi} = (\xi_1^2, \ldots, \xi_p^2)$. Under MFVB, we use the following approximation of the posterior:
\begin{equation} \label{MFVB}
q ( \bm{\beta}, \bm{\lambda}, \bm{\xi}, \sigma^2 \lvert \bm{y} ) \approx q_1^* (\bm{\beta}) q_2^* (\bm{\lambda} ), q_3^* ( \bm{\xi} ) q_4^* (\sigma^2),
\end{equation}
where 
\begin{equation} \label{MFVBeqns}
\begin{array}{l}
q_1^* ( \bm{\beta}) \sim \mathcal{N}_p \left( \bm{\beta}^*, \bm{\Sigma}^* \right), \\
q_2^* (\bm{\lambda}) \sim \displaystyle \prod_{i=1}^p \mathcal{GIG} \left( k_i^*, l^*, m^* \right), \\ 
q_3^* (\bm{\xi}) \sim \displaystyle \prod_{i=1}^p \mathcal{IG} \left( u^*, v_i^* \right), \\
q_4^* (\sigma^2) \sim \mathcal{IG} (c^*, d^*).
\end{array}
\end{equation}
From (\ref{MFVBeqns}), we can implement an efficient MFVB coordinate ascent algorithm. We optimize  the parameters,
\begin{equation*}
\left( \bm{\beta}^*, \bm{\Sigma}^*, k_1^*, \ldots, k_p^*, l^*, m^*, u, v_1^*, \ldots, v_p^*, c^*, d^* \right)
\end{equation*}
to minimize the Kullback-Leibler (KL) distance between $\pi(\bm{\beta}, \bm{\lambda}, \bm{\xi}, \sigma^2 | \bm{y})$ and $q(\bm{\beta}, \bm{\lambda}, \bm{\xi}, \sigma^2)$. Posterior inference for $\bm{\beta}$ can then carried out through the variational density $q_1^*(\bm{\beta})$. To incorporate the EM algorithm into our MFVB algorithm, we update the hyperparameters $(a, b)$ in every $k$th iteration of the coordinate ascent algorithm by solving (\ref{Mstep}) and using  $\mathbb{E}_{q_2^{* (k-1)}, a^{(k-1)}} \left[ \ln (\lambda_i^2 ) \right]$ and $\mathbb{E}_{q_3^{* (k-1)}, b^{(k-1)}} \left[ \ln (\xi_i^2) \right]$ as estimates of the summands in (\ref{Mstep}). Complete technical details for our variational EM algorithm are given in Appendix \ref{AppB} of the Supplementary Materials.

The Monte Carlo EM and variational EM algorithms are both implemented in the \textsf{R} package, \texttt{NormalBetaPrime}. In our experience, the Monte Carlo EM algorithm tends to be slower than the variational EM algorithm, but Monte Carlo EM is more accurate. This is not surprising, since MCMC converges to the exact target posterior distribution, whereas MFVB approximates the posterior with a variational density that minimizes the KL divergence. Further, the Monte Carlo EM algorithm is relatively immune to the initialization of parameters, whereas the variational EM algorithm is very sensitive to this. This is not a model-specific problem, but an inherent shortcoming of MFVB; since MFVB is optimizing a highly non-convex objective function over $O(p^2)$ parameters, it can become ``trapped'' at a suboptimal local solution. We leave the issues of deriving more efficient sampling algorithms and more accurate variational algorithms for the NBP model as problems for future research.

\subsection{Variable Selection} \label{NBPVarSelect}
Since the NBP model is absolutely continuous, it assigns zero mass to exactly sparse vectors. Therefore, selection must be performed using some posthoc method. We propose using the ``decoupled shrinkage and selection'' (DSS) method proposed by \citet{HahnCarvalho2015}. Letting $\widehat{\bm{\beta}}$ denote the posterior mean of $\bm{\beta}$,  the DSS method performs variable selection by finding the ``nearest'' exactly sparse vector to $\widehat{\bm{\beta}}$. DSS solves the optimization,
\begin{equation} \label{dss}
\widehat{\bm{\gamma}} = \displaystyle \argmin_{\bm{\gamma}} n^{-1} || \bm{X} \widehat{\bm{\beta}} - \bm{X} \bm{\gamma} ||  + \lambda || \bm{\gamma} ||_0,
\end{equation}
and chooses the nonzero entries in $\widehat{\bm{\gamma}}$ as the active set. Since (\ref{dss}) is an NP-hard combinatorial problem, \citet{HahnCarvalho2015} propose using local linear approximation, i.e. solving the following surrogate optimization problem instead:
\begin{equation} \label{dss2}
\widehat{\bm{\gamma}} = \displaystyle \argmin_{ \bm{\gamma}} n^{-1} || \bm{X} \widehat{\bm{\beta}} - \bm{X} \bm{\gamma} ||  + \lambda \displaystyle \sum_{i=1}^p  \frac{| \gamma_i |}{ | \widehat{\beta}_i |} ,
\end{equation}
where $\widehat{\beta}_i$'s are the components in the posterior mean $\widehat{\bm{\beta}}$, and $\lambda$ is chosen through 10-fold cross-validation to minimize the mean squared error (MSE). Solving this optimization is not computationally expensive, because (\ref{dss2}) is essentially an adaptive LASSO regression \cite{Zou2006} with weights $1/|\widehat{\beta_i}|, i = 1, \ldots, p$, and there exist very efficient gradient descent algorithms to find LASSO solutions, e.g. \cite{FriedmanHastieTibshirani2010}. We use the \textsf{R} package \texttt{glmnet}, developed by \citet{FriedmanHastieTibshirani2010}, to solve (\ref{dss2}). We select the nonzero entries in $\widehat{\bm{\gamma}}$ from (\ref{dss2}) as the active set of covariates. The DSS method is available for the NBP prior in the \textsf{R} package, \texttt{NormalBetaPrime}.

\section{Simulation Studies} \label{Simulations}
For our simulation studies, we implement the self-adaptive NBP model (\ref{NBPhier}) for model (\ref{y=Xbeta+epsilon}) using the Monte Carlo EM algorithm. We set $c=d=10^{-5}$ in the $\mathcal{IG}(c,d)$ prior on $\sigma^2$, while estimating $a$ and $b$ in the beta prime prior $\beta'(a, b)$ from the EM algorithm described in Section \ref{mmle}. We run the Gibbs samplers for 15,000 iterations, discarding the first 10,000 as burn-in. We use the posterior median estimator $\widehat{\bm{\beta}}$ as our point estimator and deploy the DSS strategy described in Section \ref{NBPVarSelect} for variable selection.

\subsection{Adaptivity to Different Sparsity Levels} \label{adaptivityexperiments}
In the first simulation study, we show that the self-adaptive NBP model has excellent performance under different sparsity levels. Under model (\ref{y=Xbeta+epsilon}), we generate a design matrix $\bm{X}$ where the $n$ rows are independently drawn from $\mathcal{N}_p (\bm{0}, \bm{\Gamma})$, $\mathbf{\Gamma} = (\Gamma_{ij})_{p \times p}$ with $\Gamma_{ij} = 0.5^{|i-j|}$, and then centered and scaled. We fix $\sigma^2 = 2$ and set $n = 60, p = 100$, with varying levels of sparsity: 
\begin{itemize}
\item
Experiment 1: 10 active predictors (sparse model)
\vspace{-.3cm}
\item
Experiment 2: 20 active predictors (fairly sparse model)
\vspace{-.3cm}
\item
Experiment 3: 40 active predictors (fairly dense model)
\vspace{-.3cm}
\item
Experiment 4: 60 active predictors (dense model)
\end{itemize} 
In all these settings, the true nonzero predictors in $\bm{\beta}_0$ under (\ref{y=Xbeta+epsilon}) are generated from $\mathcal{U}\left( [-2, -0.5] \cup [0.5, 2] \right)$. 

We compare the performance of the self-adaptive NBP prior with that of several other popular Bayesian and frequentist methods. For the competing Bayesian methods, we use the horseshoe \cite{CarvalhoPolsonScott2010} and the spike-and-slab LASSO  (SSL) \cite{RockovaGeorge2018}. For the horseshoe, we consider two ways of tuning the global shrinkage parameter $\tau$, which controls the sparsity of the model: 1) endowing $\tau$ with a standard half-Cauchy prior $\mathcal{C}^+ (0, 1)$,  and 2) estimating $\tau$ from restricted marginal maximum likelihood on the interval $[1/n, 1]$, as advocated by \cite{VanDerPasSzaboVanDerVaart2017}. These methods are denoted as HS-HC and HS-REML respectively. For the SSL model, the beta prior on the mixture weight $\theta$ controls the sparsity of the model. We consider two scenarios: 1) endowing $\theta$ with a $\mathcal{B}(1, p)$ prior, which induces strong sparsity, and 2) endowing $\theta$ with a $\mathcal{B}(1, 1)$ prior, which does not strongly favor sparsity. Finally, we compare the self-adaptive NBP prior to the minimax concave penalty (MCP) method \cite{Zhang2010}, the smoothly clipped absolute deviation (SCAD) method \cite{FanLi2001}, and the elastic net (ENet) \cite{ZouHastie2005}. The tuning parameters for MCP, SCAD, and EN are chosen through cross-validation. These methods are available in the \textsf{R} packages: \texttt{horseshoe}\footnote{For the HS-REML method, we slightly modified the code in the \texttt{horseshoe} function in the \texttt{horseshoe} \textsf{R} package.}, \texttt{SSLASSO}, \texttt{picasso}, and \texttt{glmnet}.

For each of our methods, we compute the mean squared error (MSE), false discovery rate (FDR), false negative rate (FNR), and overall misclassification probability (MP) averaged across 100 replications:
\begin{align*}
\begin{array}{c}
\textrm{MSE} = || \widehat{\bm{\beta}} - \bm{\beta}_0 ||_2^2/p, \hspace{.3cm} \textrm{FDR} = \textrm{FP / (TP + FP)}, \\
\textrm{FNR} = \textrm{FN / (TN + FN)}, \hspace{.3cm} \textrm{MP} = (\textrm{FP + FN}) / p, 
\end{array}
\end{align*}
where FP, TP, FN, and TN denote the number of false positives, true positives, false negatives, and true negatives respectively. 

Table \ref{Table:1} shows our results averaged across 100 replications for the NBP, HS-HC, HS-REML, SSL-$\mathcal{B}(1,p)$, SSL-$\mathcal{B}(1,1)$, MCP, SCAD, and ENet methods. Across all of the various sparsity settings, the NBP has the lowest mean squared error, indicating that it performs consistently well for estimation. In Experiments 2, 3, and 4, the NBP model also achieves either the lowest or the second lowest misclassification probability, demonstrating that it is also robust for variable selection.

The HS, SSL, MCP, and SCAD methods all perform worse as the model becomes more dense. The truncation of $\tau \in [1/n, 1]$ in the HS-REML model lowers the FDR for the horseshoe, but this also tends to overshrink large signals, leading to greater estimation error than the HS-HC model. In dense settings, endowing the sparsity parameter $\theta$ with a $\mathcal{B}(1,1)$ prior rather than a $\mathcal{B}(1,p)$ prior improves the performance under the SSL model, but not enough to be competitive with the NBP. Meanwhile, the ENet performs the worst in the sparse setting, but its performance drastically improves as the true model becomes more dense.  However, the NBP still outperforms the ENet in terms of estimation.

\subsection{More Numerical Experiments with Large $p$} \label{moreexperiments}
In this section, we consider two more settings with large $p$. In these experiments, the design matrix $\bm{X}$ is generated the same way that it was in Section \ref{adaptivityexperiments}. The active predictors are randomly selected and fixed at a certain level, and the remaining covariates are set equal to zero. 
\begin{itemize}
\item
Experiment 5: ultra-sparse model with a few large signals ($n = 100, p = 500$, 8 active predictors set equal to 5)
\item
 Experiment 6: dense model with many small signals ($n = 200, p = 400$, 200 active predictors set equal to 0.6) 
\end{itemize}

We implement Experiments 5 and 6 for the self-adaptive NBP, HS-HC, HS-REML, SSL-$\mathcal{B}(1,p)$, SSL-$\mathcal{B}(1,1)$, MCP, SCAD, and ENet models. Table \ref{Table:2} shows our results averaged across 100 replications. In Experiment 5, the NBP, HS, and SSL models all significantly outperform their frequentist competitors. The HS and SSL models do perform the best in this setting, but the NBP's performance is quite comparable to them. In particular, the NBP (as well as the HS) gives 0 for FDR, FNR, and MP, which illustrates that the NBP model is well-suited for variable selection in ultra-sparse situations. In Experiment 6, the NBP model gives the lowest MSE and lowest MP of all the methods. The ENet also performs well in this setting, but it is outperformed by the NBP in terms of estimation and variable selection. Our results for Experiment 6 confirm that the self-adaptive NBP model can effectively adapt to non-sparse situations.

Based on our numerical studies, it seems as though the horseshoe, spike-and-slab lasso, MCP, and SCAD are well-suited for sparse estimation, but cannot accommodate non-sparse situations as well. On the other hand, the elastic net seems to be a suboptimal estimator in sparse situations (e.g., in Experiment 5, its misclassification rate was 0.104, much higher than the other methods), but it greatly improves in dense settings. 

In contrast, the self-adaptive NBP prior is the most robust estimator across \textit{all} the different sparsity patterns. If the true model is very sparse, the sparsity parameter $a$ will be estimated to be very small and hence place heavier mass around zero. But if the true model is dense, the sparsity parameter $a$ will be large, so the singularity at zero disappears and the prior becomes more diffuse. As a result, small signals are more easily detected by the self-adaptive NBP prior.

\begin{table}[t!]

  \caption{Simulation results for NBP, HS-HC, HS-REML, SSL-$\mathcal{B}(1,p)$, SSL-$\mathcal{B}(1,1)$, MCP, SCAD, and ENet models, averaged across 100 replications when $n=60, p = 100$.  }
  \centering

  \medskip

\resizebox{.75\textwidth}{!}{
  \begin{tabularx}{\linewidth}{*{5}{p{.18\linewidth}}}
    \multicolumn{5}{l}{Experiment 1: sparse model (10 active predictors)} \\ \toprule
    \textbf{Method} & \textbf{MSE} & \textbf{FDR} & \textbf{FNR} & \textbf{MP} \\ \midrule
    NBP & \textbf{0.019} & 0.214 &  0.011 & 0.039 \\
    HS-HC & 0.020 & 0.128 & 0.014 & 0.029 \\
    HS-REML & 0.021  & \textbf{0.023} & 0.023 & \textbf{0.023} \\
   SSL-$\mathcal{B}(1,p)$ & 0.020 & 0.066  & 0.019 & 0.026 \\ 
    SSL-$\mathcal{B}(1,1)$ & 0.025  & 0.151 & 0.017 & 0.036 \\
    MCP & 0.020 & 0.238 & \textbf{0.014} & 0.046 \\
    SCAD & 0.028 & 0 &  0.1 & 0.1 \\
   ENet & 0.037 & 0.730 & 0.006 & 0.284 \\ \bottomrule
  \end{tabularx}}

  \medskip
\resizebox{.75\textwidth}{!}{
  \begin{tabularx}{\linewidth}{*{5}{p{.18\linewidth}}}
    \multicolumn{5}{l}{Experiment 2: fairly sparse model (20 active predictors)} \\ \toprule
    \textbf{Method} & \textbf{MSE} & \textbf{FDR} & \textbf{FNR} & \textbf{MP} \\ \midrule
    NBP & \textbf{0.077} & 0.202 & 0.050 & 0.083 \\
    HS-HC &  0.110 & 0.235 & 0.084 & 0.11 \\
    HS-REML & 0.286  & \textbf{0.130} & 0.115 & 0.119 \\
   SSL-$\mathcal{B}(1,p)$ & 0.090 & 0.175  & 0.053 & \textbf{0.078} \\  
    SSL-$\mathcal{B}(1,1)$ & 0.090  & 0.222 & 0.048 & 0.086 \\
    MCP & 0.238 & 0.321 & 0.091 & 0.142 \\
    SCAD & 0.226 & 0.791 &  0.199 & 0.252 \\
   ENet & 0.096 & 0.610 & \textbf{0.031} & 0.310 \\ \bottomrule
  \end{tabularx}}

  \medskip
\resizebox{.75\textwidth}{!}{
  \begin{tabularx}{\linewidth}{*{5}{p{.18\linewidth}}}
    \multicolumn{5}{l}{Experiment 3: fairly dense model (40 active predictors)} \\ \toprule
    \textbf{Method} & \textbf{MSE} & \textbf{FDR} & \textbf{FNR} & \textbf{MP} \\ \midrule
    NBP & \textbf{0.448} & 0.251 & \textbf{0.240} & \textbf{0.246} \\
    HS-HC &  0.535 & 0.243 & 0.256 & 0.254 \\
    HS-REML & 1.10  & \textbf{0.233} & 0.338 & 0.325 \\
   SSL-$\mathcal{B}(1,p)$ & 0.728 & 0.300  & 0.270 & 0.279 \\  
    SSL-$\mathcal{B}(1,1)$ & 0.665  & 0.308 & 0.260 & 0.276 \\
    MCP & 1.31 & 0.298 & 0.343 & 0.344 \\
    SCAD & 1.21 & 0.604 &  0.401 & 0.440 \\
   ENet & 0.453 & 0.423 & 0.198 & 0.320 \\ \bottomrule
  \end{tabularx}}

  \medskip
\resizebox{.75\textwidth}{!}{
  \begin{tabularx}{\linewidth}{*{5}{p{.18\linewidth}}}
    \multicolumn{5}{l}{Experiment 4: dense model (60 active predictors)} \\ \toprule
    \textbf{Method} & \textbf{MSE} & \textbf{FDR} & \textbf{FNR} & \textbf{MP} \\ \midrule
    NBP & \textbf{0.760} & 0.173 & 0.467 & 0.344 \\
    HS-HC &  1.10 & 0.184 & 0.495 & 0.395 \\
    HS-REML & 1.76 & \textbf{0.149} & 0.552 & 0.489 \\
   SSL-$\mathcal{B}(1,p)$ & 1.53 & 0.223  & 0.506 & 0.409 \\ 
    SSL-$\mathcal{B}(1,1)$ & 1.40 & 0.226 & 0.495 & 0.395 \\
    MCP & 1.31 & 0.298 & \textbf{0.343} & 0.359 \\
    SCAD & 2.18 & 0.430 &  0.603 & 0.589 \\
   ENet & 0.892 & 0.260 & 0.426 & \textbf{0.336} \\ \bottomrule
  \end{tabularx}}

  \label{Table:1}
\end{table}

\begin{table}[t!]

  \caption{More simulation results for NBP, HS-HC, HS-REML, SSL-$\mathcal{B}(1,p)$, SSL-$\mathcal{B}(1,1)$, MCP, SCAD, and ENet models, averaged across 100 replications.  }

  \centering

  \medskip

\resizebox{.8\textwidth}{!}{
  \begin{tabularx}{\linewidth}{*{5}{p{.18\linewidth}}}
    \multicolumn{5}{l}{Experiment 5: $n = 100$, $p = 500$, 8 active predictors set equal to 5. } \\ \toprule
    \textbf{Method} & \textbf{MSE} & \textbf{FDR} & \textbf{FNR} & \textbf{MP} \\ \midrule
    NBP & 0.0007 & \textbf{0} &  \textbf{0} & \textbf{0} \\
    HS-HC & \textbf{0.0005} & \textbf{0} & \textbf{0} & \textbf{0} \\
    HS-REML & \textbf{0.0005}  & \textbf{0} & \textbf{0} & \textbf{0} \\
   SSL-$\mathcal{B}(1,p)$ & \textbf{0.0005} & 0.037  & \textbf{0} & 0.0007 \\ 
    SSL-$\mathcal{B}(1,1)$ & 0.0008  & 0.089 & \textbf{0} & 0.0017 \\
    MCP & 0.078 & 0.124 & 0.0012 & 0.011 \\
    SCAD & 0.081 & 0.984 &  0.016 & 0.031 \\
   ENet & 0.067 & 0.859 & \textbf{0} & 0.104 \\ \bottomrule
  \end{tabularx}}

  \medskip
\resizebox{.8\textwidth}{!}{
  \begin{tabularx}{\linewidth}{*{5}{p{.18\linewidth}}}
    \multicolumn{5}{l}{Experiment 6: $n =200$, $p = 400$, 200 active predictors set equal to 0.6} \\ \toprule
    \textbf{Method} & \textbf{MSE} & \textbf{FDR} & \textbf{FNR} & \textbf{MP} \\ \midrule
    NBP & \textbf{0.031} & 0.273 & 0.400 & \textbf{0.351} \\
    HS-HC &  0.041 & 0.261 & 0.423 & 0.384 \\
    HS-REML & 0.049  & \textbf{0.204} & 0.469 & 0.444 \\
   SSL-$\mathcal{B}(1,p)$ & 0.095 & 0.311  & 0.462 & 0.437 \\  
    SSL-$\mathcal{B}(1,1)$ & 0.093  & 0.334 & 0.458 & 0.433 \\
    MCP & 0.058 & 0.213 & 0.479 & 0.462 \\
    SCAD & 0.051 & 0.488 &  0.499 & 0.498 \\
   ENet & 0.038 & 0.346 & 0.362 & 0.355 \\ \bottomrule
  \end{tabularx}}

  \label{Table:2}
\end{table}

\section{Analysis of a Gene Expression Data Set} \label{DataAnalysis}
We analyze a real data set from a study on Bardet-Biedl syndrome (BBS), an autosomal recessive disorder which leads to progressive vision loss and which is caused by a mutation in the TRIM32 gene. This data set was first studied by \citet{Scheetz2006} and is available in the \textsf{R} package \texttt{flare}. This data set contains $n = 120$ samples with TRIM32 as the response variable and the expression levels of $p = 200$ other genes as the covariates. 
	
To determine TRIM32's association with these other genes, we implement the self-adaptive NBP, HS-HC, HS-REML, SSL-$\mathcal{B}(1,p)$, SSL-$\mathcal{B}(1,1)$, MCP, SCAD, and ENet models on this data set after centering and scaling $\bm{X}$ and $\bm{y}$. To assess these methods' predictive performance, we perform five-fold cross validation, using 80 percent of the data as our training set to obtain an estimate of $\bm{\beta}$, $\widehat{\bm{\beta}}_{\textrm{train}}$. We then use $\widehat{\bm{\beta}}_{\textrm{train}}$ to compute the mean squared error of the residuals on the remaining 20 percent of the left-out data. We repeat this five times, using different training and test sets each time, and take the average MSE as our mean squared prediction error (MSPE). 

Table \ref{Table:3} shows the results of our analysis. The NBP and ENet models give the best predictive performance of all the methods, with 31 genes and 26 genes selected as significantly associated with TRIM32, respectively. The ENet has slightly lower MSPE, but the NBP model's performance is very similar to the ENet's. The HS, SSL, MCP, and SCAD methods result in the most parsimonious models, with 6 or fewer genes selected, but their average prediction errors are all higher than the NBP's or the ENet's.

\begin{table}[t!]
  \caption{Results for analysis of the Bardet-Biedl syndrome (BBS) dataset.}
  \centering
\resizebox{.8\textwidth}{!}{
\begin{tabular}{l*{3}{c}r}
\hline
\textbf{Method}  & \textbf{Number of Genes Selected} & \textbf{MSPE} \\
\hline
NBP & 31 & 0.466 \\
HS-HC & 6 & 0.797 \\
HS-REML & 4 & 0.616 \\
SSL-$\mathcal{B}(1,p)$ & 3 & 0.594 \\
SSL-$\mathcal{B}(1,1)$ & 3 & 0.504 \\
MCP & 5 & 0.582 \\
SCAD & 5 & 0.603 \\
ENet & 26 & \textbf{0.462} \\
\hline
\end{tabular}}
  \label{Table:3}
\end{table}

Figure \ref{fig:4} plots the posterior 95 percent credible intervals for the 31 genes that the NBP model selected as significant. Figure \ref{fig:4} shows that the self-adaptive NBP prior is able to detect small gene expression values that are very close to zero.  On this particular data set, the slightly more dense models had much better prediction performance than the most parsimonious models, suggesting that there may be a number of small signals in our data.

\section{Concluding Remarks and Future Work}

\begin{figure}[t!]
\centering
\includegraphics[scale=0.52, trim={0 0.75cm 0 1.5cm}]{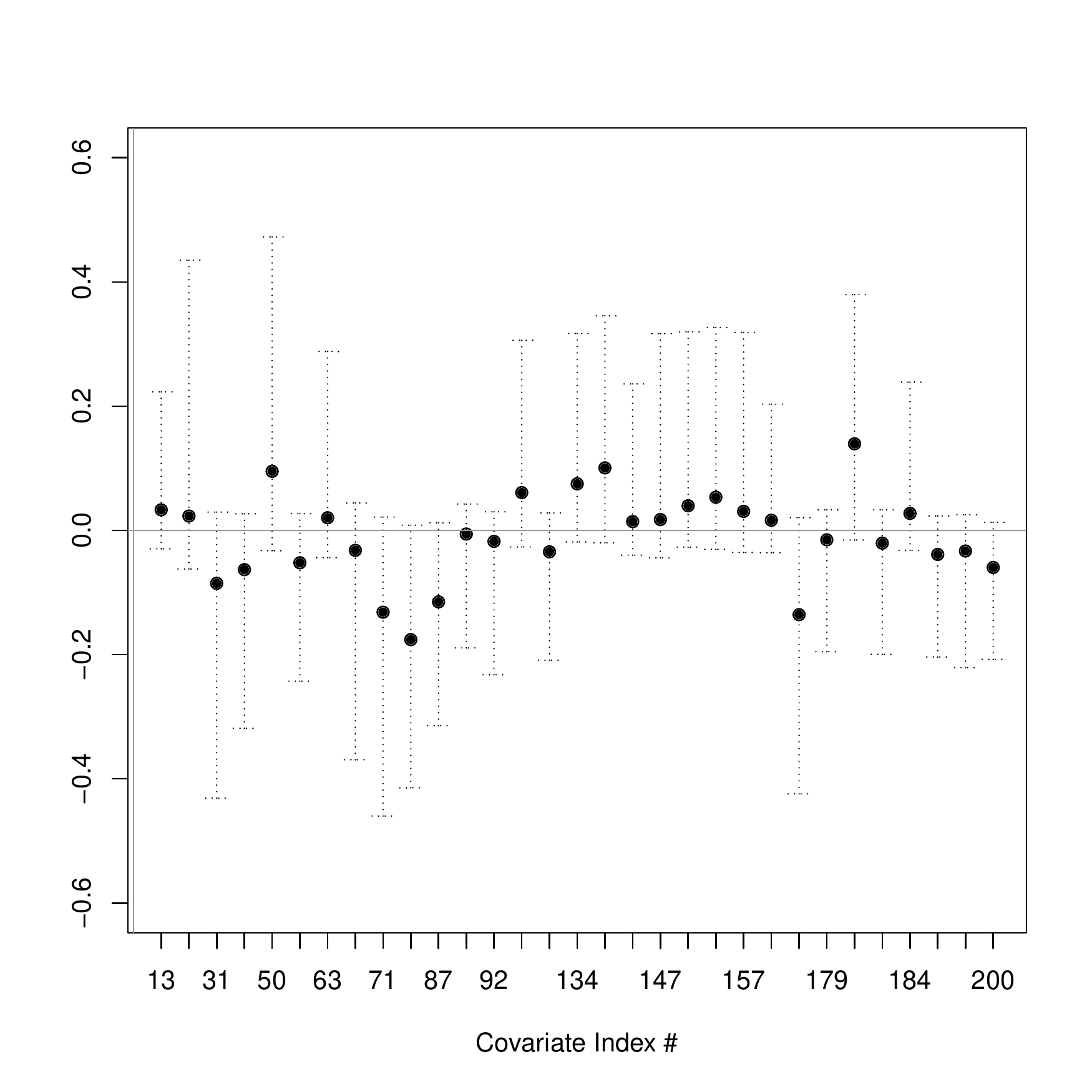} 
\caption{Posterior median and 95\% credible intervals for the 31 genes that were selected as significantly associated with TRIM32 by the NBP model.  }
\label{fig:4}
\end{figure}

In this paper, we have introduced the normal-beta prime (NBP) model for high-dimensional Bayesian linear regression. We proved that the NBP prior obtains the (near) minimax posterior contraction rate in the asymptotic regime where $p \gg n$ and the underlying model is sparse. To make our prior self-adaptive in finite samples, we introduced an empirical Bayes approach for estimating the NBP's hyperparameters through maximum marginal likelihood (MML). Our MML approach for estimating hyperparameters affords the NBP a great deal of flexibility and adaptivity to different levels of sparsity and different signal strengths, while avoiding degeneracy. 

Future work will be to extend the NBP prior to more complex and more flexible models, such as nonparametric regression or semiparametric regression with unknown error distribution. The NBP prior can also be employed for other statistical problems like density estimation or classification. We conjecture that due to its flexibility, the NBP prior would retain its strong empirical and theoretical properties in these other settings.

Additionally, we would like to provide further theoretical support for the marginal maximum likelihood approach described in Section \ref{mmle}. Although there are philosophical reasons for MML (namely, that it maximizes the ``model evidence''), we wish to investigate if the MML estimates of $(a, b)$ also lead to (near) minimax posterior contraction under the same conditions as those in Section \ref{minimaxregcontraction}. Currently, theoretical justifications for MML under model (\ref{y=Xbeta+epsilon}) have been confined to either the simple normal means model ($\bm{X} = \bm{I}$, $n=p$) or the scenario where $p \leq n$ and the MML estimate can be explicitly calculated in closed form (as is the case for the hyperparameter $g$ in $g$-priors). See, e.g., \cite{VanDerPasSzaboVanDerVaart2017, JohnstoneSilverman2004, GeorgeFoster2000, SparksKhareGhosh2015}. Recently, \citet{RousseauSzabo2017} extended the class of models for which the posterior contraction rate can be obtained under MML estimates of a hyperparameter in the prior, but unfortunately, their framework does not seem to be applicable to the high-dimensional linear regression model (\ref{y=Xbeta+epsilon}), which is complicated by the presence of a high-dimensional, rank-deficient design matrix $\bm{X}$. 

We note that several other papers, e.g. \cite{MartinWalker2014, MartinMessWalker2017}, have derived theoretical results under empirical Bayes methods, but the approaches in these papers are not based on marginal maximum likelihood. Instead, their methods are ``empirical Bayes'' in the sense that the prior is constructed from the data, while other hyperparameters are specified \textit{a priori} based on asymptotic arguments. We hope to address the theoretical aspects of the self-adaptive NBP model with MML-estimated hyperparameters in future work.

\bibliographystyle{apa}
\bibliography{Reference}

\begin{appendix}

\section{Proofs of Main Theorems} \label{appA}

Before proving Theorem \ref{Th:3.1}, we restate the main results on posterior consistency from \citet{SongLiang2017}. Proposition \ref{Prop:B.1} is a restatement of Theorems A.1 and A.2 in \cite{SongLiang2017}.

\begin{proposition}
\label{Prop:B.1}
Consider the linear regression model (\ref{truemodel}) and suppose that condition \ref{As:A1}-\ref{As:A5} hold. Suppose that the prior for $\pi(\bm{\beta}, \sigma^2)$ is of the form,
\begin{equation} \label{priorspec}
\pi(\bm{\beta} | \sigma^2) = \displaystyle \prod_{i=1}^p \left[ g(\beta_i / \sigma) / \sigma \right], \hspace{.3cm} \sigma^2 \sim \mathcal{IG}(c, d).
\end{equation}
Suppose $r_n = M \sqrt{s_n \log p_n / n}$, where $M>0$ is sufficiently large. If the density $g(\cdot)$ in (\ref{priorspec}) satisfies
\begin{equation} \label{contractionconditions}
\begin{array}{c}
1 - \displaystyle \int_{-k_n}^{k_n} g(x) dx \leq p_n^{-(1+u)}, \\
- \log \left( \displaystyle \inf_{x \in [-E_n, E_n]} g(x) \right) = O(\log p_n),
\end{array}
\end{equation}
where $u > 0$ is a constant and $k_n \asymp \sqrt{s_n \log p_n / n} / p_n$, then the following results hold:

\begin{equation*}
\begin{array}{l}
\Pr_{\bm{\beta}_0} \left( \Pi \left( \bm{\beta}: || \bm{\beta} - \bm{\beta}_0 ||_2 \geq c_1 \sigma_0 r_n | \bm{y}_n \right) \geq e^{-c_2 n r_n^2} \right) \leq e^{-c_3 n r_n^2 }, \\
\Pr_{\bm{\beta}_0} \left( \Pi \left( \bm{\beta}: || \bm{\beta} - \bm{\beta}_0 ||_1 \geq c_1 \sigma_0 \sqrt{s_n} r_n | \bm{y}_n \right) \geq e^{-c_2 n r_n^2} \right) \leq e^{-c_3 n r_n^2 },  \\
\Pr_{\bm{\beta}_0} \left( \Pi \left( \bm{\beta}: || \bm{X}_n \bm{\beta} - \bm{X}_n \bm{\beta}_0 ||_2 \geq c_0 \sigma _0 \sqrt{n} r_n | \bm{y}_n \right) < 1 -  e^{-c_2 n r_n^2} \right) \leq e^{-c_3 n r_n^2 },  \\
\Pr_{\bm{\beta}_0} \left( \Pi \left( \bm{\beta}: \textrm{at least } \widetilde{q}_n \textrm{ entries of } | \bm{\beta}/\sigma | \textrm{ are larger than } k_n | \bm{y}_n \right) > e^{-c_2 n r_n^2} \right) \leq e^{-c_3 n r_n^2},
\end{array}
\end{equation*}
for some constants $c_0, c_1, c_2, c_3 > 0$, and $\widetilde{q}_n \asymp s_n$.
\end{proposition}

Before proving Theorem \ref{Th:3.1}, we also prove the following two lemmas.
\begin{lemma} \label{Lemma1forTh3.1}
Suppose that $a_n \rightarrow 0$ as $n \rightarrow \infty$ and $b \in (1, \infty)$ as $n \rightarrow \infty$. Then
\begin{equation} \label{normalizingconstantrelation}
\frac{\Gamma(a_n+b)}{\Gamma(a_n) \Gamma(b) } \asymp a_n.
\end{equation}
\end{lemma}
\begin{proof} [Proof of Lemma \ref{Lemma1forTh3.1}]
Rewrite (\ref{normalizingconstantrelation}) as
\begin{align*} \label{normalizingconstant1}
\frac{\Gamma(a_n+b)}{\Gamma(a_n) \Gamma(b) }  & = \frac{a_n \Gamma(a_n+b+1)}{(a_n+b) \Gamma(a_n+1) \Gamma(b)} \\
& = \frac{a_n}{a_n+b} \left( \frac{1}{ \int_{0}^{1} u^{a_n} (1-u)^{b-1} du} \right). \numbereqn
\end{align*}
We have the following inequalities:
\begin{equation} \label{normalizingconstant2}
\int_{0}^{1} u^{a_n} (1-u)^{b-1} du \leq \int_{0}^{1} (1-u)^{b-1} du =b^{-1},
\end{equation}
and
\begin{align*} \label{normalizingconstant3}
\int_{0}^{1} u^{a_n} (1-u)^{b-1} du & \geq \displaystyle \int_{1/2}^{1} u^{a_n} (1-u)^{b-1} du \\
& \geq 2^{-a_n} \int_{1/2}^{1} (1-u)^{b-1} du \\
& = 2^{-a_n} 2^{-b} b^{-1}. \numbereqn
\end{align*}
Thus, from (\ref{normalizingconstant1})-(\ref{normalizingconstant3}), we have
\begin{equation} \label{normalizingconstant4}
\frac{a_n b}{a_n+b}  \leq \frac{\Gamma(a_n+b)}{\Gamma(a_n) \Gamma(b)} \leq \frac{a_n  2^{a_n+b} b }{a_n+b}.
\end{equation}
Since $a_n \rightarrow 0$ as $n \rightarrow \infty$, we have $b/(a_n+b) \sim 1$ and $2^{a_n+b} b / (a_n+b) \sim 2^b$, and thus, from (\ref{normalizingconstant4}), we have $\Gamma(a_n + b) / \Gamma(a_n) \Gamma(b) \asymp a_n$ as $n \rightarrow \infty$.
\end{proof}

\begin{lemma} \label{Lemma2forTh3.1}
Let $b>1$. Then for any $a > 0$, $\beta '(a, b)$ is stochastically dominated by $\beta ' (a, 1)$.
\end{lemma}
\begin{proof}[Proof of Lemma \ref{Lemma2forTh3.1}]
Let $f(x | a, b)$ denote the probability density function (pdf) for the beta prime density, $\beta ' (a, b)$. We have
\begin{align*}
\frac{f(x | a, 1)}{f( x | a, b)} \propto \frac{ x^{a-1} (1+x)^{-a-1} }{ x^{a-1} (1+x)^{-a-b} } = (1+x)^{b-1},
\end{align*}
which is increasing in $x$ due to our assumption that $b > 1$. Hence, by the monotone likelihood ratio property, $\beta ' (a,b)$ is stochastically dominated by $\beta ' (a, 1)$ for any $b > 1$.
\end{proof}

\begin{proof}[Proof of Theorem \ref{Th:3.1}]
By Proposition \ref{Prop:B.1}, it is sufficient to verify that the NBP prior for each coefficient $\pi(\beta_i ), i = 1, \ldots, p_n,$ satisfies the two conditions (\ref{contractionconditions}). We first verify the first condition. Let $g( \cdot )$ be the marginal pdf of $\pi(\beta )$ for a single coefficient $\beta$. The pdf $g(x)$ under the NBP prior is
\begin{equation} \label{NBPmarginal}
g(x) = \frac{\Gamma(a_n+b)}{(2 \pi)^{1/2} \Gamma(a_n) \Gamma(b)} \displaystyle \int_{0}^{\infty} \exp \left( - \frac{x^2}{2 \omega^2} \right) (\omega^2)^{a_n-3/2} (1+ \omega^2)^{-a_n-b} d \omega^2. 
\end{equation}
By the symmetry of $g(x)$ and Fubini's Theorem, we have from (\ref{NBPmarginal}) that
\begin{align*} \label{firstcondition1}
& 1 - \displaystyle \int_{-k_n}^{k_n} g(x) dx = 2 \displaystyle \int_{k_n}^{\infty} g(x) dx \\
& \qquad = \frac{2 \Gamma(a_n+b)}{(2 \pi)^{1/2} \Gamma(a_n) \Gamma(b)} \displaystyle \int_{k_n}^{\infty} \displaystyle \int_{0}^{\infty} \exp \left( - \frac{x^2}{2 \omega^2} \right) (\omega^2)^{a_n-3/2} (1+ \omega^2)^{-a_n-b} d \omega^2 dx \\ 
& \qquad = \frac{\Gamma(a_n+b)}{\Gamma(a_n) \Gamma(b)} \displaystyle \int_{0}^{\infty} (\omega^2)^{a_n-1} (1+ \omega^2)^{-a_n-b} \left[ 2  \displaystyle \int_{k_n}^{\infty} (2 \pi \omega^2)^{-1/2} \exp \left( - \frac{x^2}{2 \omega^2} \right) dx \right] d \omega^2  \numbereqn
\end{align*}
Letting $X \sim \mathcal{N}(0, \omega^2)$, we see the inner integral in (\ref{firstcondition1}) is $\Pr (|X| \geq k_n)$. We use the concentration inequality, $\Pr ( |X| \geq k_n) \leq 2 e^{-k_n^2 / 2 \omega^2}$, to further bound (\ref{firstcondition1}) above as
\begin{align*} \label{firstcondition2}
2 \displaystyle \int_{k_n}^{\infty} g(x) dx & \leq \frac{2 \Gamma(a_n+b)}{\Gamma(a_n) \Gamma(b)} \displaystyle \int_{0}^{\infty} (\omega^2)^{a_n-1} (1+ \omega^2)^{-a_n-b} e^{-k_n^2 / 2 \omega^2} d \omega^2   \\
& \leq 2 a_n \displaystyle \int_{0}^{\infty} (\omega^2)^{a_n-1} (1+ \omega^2)^{-a_n-1} e^{-k_n^2 / 2 \omega^2} d \omega^2 \\
& = 2 a_n \displaystyle \int_{0}^{\infty} (1+u)^{-a_n-1} e^{-u(k_n^2/2)} du \\
& \leq 2 a_n \displaystyle \int_{0}^{\infty} e^{-u(k_n^2/2)} du \\
& = \frac{4 a_n}{k_n^2} \\
& \lesssim p_n^{-(1+u)}, \numbereqn
\end{align*}
where we used the assumption that $b \in (1, \infty)$ and Lemma \ref{Lemma2forTh3.1} in the second inequality, a transformation of variables $u = 1/ \omega^2$ in the first equality, and the assumption that  $a_n \lesssim k_n^2 p_n^{-(1+u)}$ for the final inequality of the above display. Thus, combining (\ref{firstcondition1})-(\ref{firstcondition2}) shows that the first condition in (\ref{contractionconditions}) holds.

We now show that the second condition of (\ref{contractionconditions}) also holds under our assumptions on $(a_n, b)$ and our assumption on the rate of growth on $E_n$ in \ref{As:A5}. With a change of variables, $z = x^2 / 2 \omega^2$, in (\ref{NBPmarginal}), we can rewrite the marginal pdf of the NBP prior, $g(x)$, as
\begin{equation} \label{NBPgx}
g(x) = \frac{\Gamma(a_n+b)}{2^{1-b} \pi^{1/2} \Gamma(a_n) \Gamma(b)} (x^2)^{a_n-1/2} \displaystyle \int_{0}^{\infty} e^{-z} z^{b - 1/2} (x^2+2z)^{-a_n-b} dz.
\end{equation}
By the symmetry of $g(x)$, the infimum of $g(x)$ on the interval $[-E_n, E_n]$ occurs at either $-E_n$ or $E_n$. From (\ref{normalizingconstantrelation}) in Lemma \ref{Lemma1forTh3.1}, (\ref{NBPgx}), and the assumptions that $E_n$ is nondecreasing and $b \in (1, \infty)$, we have
\begin{align*} \label{secondcondition1}
 \displaystyle \inf_{x \in [-E_n, E_n] } g(x) & \gtrsim a_n (E_n^2)^{a_n-1/2} \displaystyle \int_{0}^{\infty} e^{-z} z^{b-1/2} (E_n^2+2z)^{-a_n-b} dz \\
& = a_n (E_n^2)^{a_n - 1/2} \displaystyle \int_{0}^{\infty} e^{-z} \left( \frac{z}{E_n^2+2z} \right)^{b-1/2} \left( \frac{1}{E_n^2+2z} \right)^{a_n+1/2} dz \\
& \geq a_n (E_n^2)^{a_n - 1/2} \displaystyle \int_{1}^{2} e^{-z} \left( \frac{z}{E_n^2+2z} \right)^{b-1/2} \left( \frac{1}{E_n^2+2z} \right)^{a_n+1/2} dz \\
& \gtrsim a_n (E_n^2)^{a_n-1/2} (E_n^2+2)^{-b+1/2} (E_n^2 + 4)^{-a_n-1/2}   \\
& \asymp a_n (E_n^2)^{-b-1/2}. \numbereqn
\end{align*}
By assumption, $a_n \lesssim k_n^2 p_n^{-(1+u)}$ for some $u > 0$, and $\log(E_n) = O(\log p_n)$. Therefore, it follows from (\ref{secondcondition1}) that
\begin{align*} \label{secondcondition2}
- \log \left(  \displaystyle \inf_{x \in [-E_n, E_n]} g(x) \right) & \lesssim - \log(k_n^2 p_n^{-(1+u)}) + (b+1/2) \log p_n \\
& \lesssim - \log (p_n^{-(3+u)}) + (b+1/2) \log p_n \\
& \lesssim \log p_n, \numbereqn
\end{align*} 
where we used the fact that $k_n \asymp \sqrt{s_n \log p_n / n}/p_n$ and Assumption \ref{As:A4} that $s_n = o(n / \log p_n)$, and so $k_n \lesssim p_n^{-1}$. Thus, the second condition in (\ref{contractionconditions}) also holds. 

We have shown that as long as $a_n \lesssim k_n^2 p_n^{-(1+u)}$ for some $u > 0$, $b \in (1, \infty)$, and $\log(E_n) = O( \log p_n)$ in Assumption \ref{As:A5}, the two conditions (\ref{contractionconditions}) in Proposition \ref{Prop:B.1} are satisfied. Hence, Theorem \ref{Th:3.1} has been proven.
\end{proof}

We now prove that the M-step of the EM algorithm for obtaining the MML estimators of $(a, b)$ in Section \ref{mmle} always has a unique solution where $a^{(k)}>0, b^{(k)}>0$ at every $k$th iteration, and therefore, our EM algorithm avoids collapse to zero.

\begin{proof}[Proof of Theorem \ref{Th:4.1}]
At the $k$th iteration of the EM algorithm, the $(a, b)$ that solves (\ref{Mstep}) is
\begin{equation} \label{digammasystem}
\begin{array}{rclr}
 \psi(a) & = &  \frac{1}{p} \displaystyle \sum_{i=1}^p U_i ( \lambda_i^2) , &  a \geq 0, \\
 \psi(b)  & = & - \frac{1}{p} \displaystyle \sum_{i=1}^p , V_i ( \xi_i^2), & b \geq 0,
\end{array}
\end{equation}
where $U_i (\lambda_i^2)$ is an estimate of $\mathbb{E}_{a^{(k-1)}} \left[ \log(\lambda_i^2) | \bm{y} \right]$ and $V_i (\xi_i^2)$ is an estimate of $\mathbb{E}_{b^{(k-1)}} \left[ \log (\xi_i^2) | \bm{y} \right]$ taken from either the Gibbs sampler or the MFVB coordinate ascent algorithm. Since the $\lambda_i$'s and $\xi_i$'s, $i = 1, \ldots, p$, are strictly greater than zero and are drawn from $\mathcal{GIG}$ and $\mathcal{IG}$ densities in the Gibbs sampling algorithm or the MFVB algorithm (and thus, expectations of $\log(\lambda_i^2)$ and $\log(\xi_i^2)$, $i=1, \ldots, p$, are well-defined and finite), $U_i$ and $V_i$, $i = 1, \ldots, p$, exist and are finite. 

The digamma function $\psi(x)$ is continuous and monotonically increasing for all $x \in (0, \infty)$, with a range of $(-\infty, \infty)$ on the domain of positive reals. Therefore, for any $y \in \mathbb{R}$, there exists a unique $x \in (0, \infty)$ so that $\psi(x) = y$.  Since we impose the constraint that $ a \geq 0$, there must be a unique $\widehat{a}^{(k)} > 0$ that solves the first equation in (\ref{digammasystem}). Similarly, there exists a unique $\widehat{b}^{(k)} > 0$ that solves the second equation in (\ref{digammasystem}).  
\end{proof}
 
\section{Technical Details of the Monte Carlo EM and Variational EM Algorithms for the Self-Adaptive NBP Model} \label{AppB}

\subsection{Monte Carlo EM Algorithm} \label{MCMCforNBP}
After initializing $(\bm{\beta}, \lambda_1, \ldots, \lambda_p, \xi_1, \ldots, \xi_p, \sigma^2)$, we iteratively cycle through sampling from the full conditional densities in (\ref{NBPconditionals}). As described in Section \ref{PosteriorSimulation}), we incorporate the EM algorithm for obtaining the MML estimates of $(a, b)$ by solving for $(a, b)$ in (\ref{Mstep}) every $M = 100$ iterations of the Gibbs sampler. To assess convergence, we compute the square of the Euclidean distance between $(\widehat{a}^{(k-1)}, \widehat{b}^{(k-1)})$ and $( \widehat{a}^{(k)}, \widehat{b}^{(k)})$ at the $k$th iteration of the EM Monte Carlo algorithm, and if it falls below a small tolerance $\delta > 0$, then we set our MML estimates as $(\widehat{a}, \widehat{b}) = (\widehat{a}^{(k)}, \widehat{b}^{(k)} )$ and draw a final sample from the Gibbs sampler. 

We recommend setting $\delta = 10^{-6}$. If the square of the $\ell_2$ distance has not fallen below $\delta$ after a large number of iterations (we use a maximum number of 100 iterations for the EM algorithm, so that 10,000 total iterations of the Gibbs sampler have been sampled at this point), then we terminate the EM algorithm and use the final estimate from the 100th iteration as $(\widehat{a}, \widehat{b})$. In our experience, even if the square of the $\ell_2$ distance between $(\widehat{a}^{(k-1)}, \widehat{b}^{(k-1)})$ and $( \widehat{a}^{(k)}, \widehat{b}^{(k)})$  does not quite fall underneath the small $\delta > 0$ after $k=100$ updates, the successive iterates are still very close to one another at this point. Thus, all these later estimates of $(a, b)$ would have a similar effect on posterior inference.   Algorithm 1 at the end of Section \ref{AppB} gives the complete steps for implementing the EM/Gibbs algorithm for our model.

\subsection{Variational EM Algorithm} \label{VBforNBP}
Let $\bm{\lambda} = (\lambda_1^2, \ldots, \lambda_p^2)$ and $\bm{\xi} = (\xi_1^2, \ldots, \xi_p^2)$ from (\ref{NBPconditionals}). The mean field variational Bayes (MFVB) approach stems from the following lower bound:
\begin{equation} \label{VBlowerbound}
 \log \pi(\bm{y}) \geq \displaystyle \int_{(\bm{\beta}, \bm{\lambda}, \bm{\xi}, \sigma^2)} q(\bm{\beta}, \bm{\lambda}, \bm{\xi}, \sigma^2) \log \left( \frac{\pi(\bm{y}, \bm{\beta}, \bm{\lambda}, \bm{\xi}, \sigma^2, \bm{\gamma})}{q(\bm{\beta}, \bm{\lambda}, \bm{\xi}, \sigma^2)} \right) d(\bm{\beta}, \bm{\lambda}, \bm{\xi}, \sigma^2 ) \equiv \mathcal{L}[q(\cdot)],
\end{equation}
where $\mathcal{L}[q(\cdot)]$ is known as the evidence lower bound (ELBO). Here, we constrain $q(\bm{\beta}, \bm{\lambda}, \bm{\xi}, \sigma^2) = q_1^* (\bm{\beta}) q_2^* (\bm{\lambda}) q_3^* ( \bm{\xi}) q_4^* (\sigma^2)$  and the $q_i$'s, $i = 1, \ldots, 4$, to be families that ensure that (\ref{VBlowerbound}) is tractable. This approach is also known as mean field approximation (MFVB). The parameters in $q_1^*$,  $q_2^*$,  $q_3^*$, and $q_4^*$ are then found by maximizing (\ref{VBlowerbound}), which is equivalent to minimizing the Kullback-Leibler (KL) distance between $\pi(\bm{\beta}, \bm{\lambda}, \bm{\xi}, \sigma^2 | \bm{y})$ and $q(\bm{\beta}, \bm{\lambda}, \bm{\xi}, \sigma^2)$. Due to independence, $\pi(\bm{\beta} | \bm{y} )$ can be approximated by $q_1^*(\bm{\beta})$ and posterior inference can be carried out through $q_1^* (\bm{\beta})$. For a detailed review of variational inference, see \citet{BleiKucukelbirMcAuliffe2017}.

Based on the full conditional densities in (\ref{PosteriorSimulation}), we use the approximation,
\begin{equation*} 
q ( \bm{\beta}, \bm{\lambda}, \bm{\xi}, \sigma^2 \lvert \bm{y} ) \approx q_1^* (\bm{\beta}) q_2^* (\bm{\lambda} ), q_3^* ( \bm{\xi} ) q_4^* (\sigma^2),
\end{equation*}
where 
\begin{equation*} 
\begin{array}{l}
q_1^* ( \bm{\beta}) \sim \mathcal{N}_p \left( \bm{\beta}^*, \bm{\Sigma}^* \right), \\
q_2^* (\bm{\lambda}) \sim \displaystyle \prod_{i=1}^p \mathcal{GIG} \left( k_i^*, l^*, m^* \right), \\ 
q_3^* (\bm{\xi}) \sim \displaystyle \prod_{i=1}^p \mathcal{IG} \left( u^*, v_i^* \right), \\
q_4^* (\sigma^2) \sim \mathcal{IG} (c^*, d^*),
\end{array}
\end{equation*}
and
\begin{equation} \label{MFVBeqns2}
\begin{array}{c}
\bm{\beta}^* = \left( \bm{X}^\top \bm{X} + \bm{D}^* \right)^{-1} \bm{X}^\top \bm{y}, \hspace{.3cm} \bm{\Sigma}^* =  \mathbb{E}_{q_4^*} (\sigma^2) \left( \bm{X}^\top \bm{X} + \bm{D}^* \right)^{-1} , \\
\bm{D}^* = \textrm{diag} \left( \mathbb{E}_{q_2^*} ( \lambda_1^{-2}) \mathbb{E}_{q_3^*} ( \xi_1^{-2}), \ldots, \mathbb{E}_{q_2^*} ( \lambda_p^{-2}) \mathbb{E}_{q_3^*} (\xi_p^{-2}) \right),  \\
k_i = \mathbb{E}_{q_1^*} ( \beta_i^2 ) \mathbb{E}_{q_4^*} (\sigma^{-2}) \mathbb{E}_{q_3^*} (\xi_i^{-2}), i=1, \ldots, p, \hspace{.3cm} l^* = 2, \hspace{.3cm} m^*=  a - \frac{1}{2}, \\  
u^* = b + \frac{1}{2}, \hspace{.3cm} v_i^* = \frac{1}{2} \mathbb{E}_{q_1^*} ( \beta_i^2 ) \mathbb{E}_{q_4^*}(\sigma^{-2}) \mathbb{E}_{q_2^*} (\lambda_i^{-2}) + 1, i=1, \ldots, p, \\ 
c^* = \frac{n+p+2c}{2}, \hspace{.3cm} d^* = \frac{ \mathbb{E}_{q_1^*} ( || \bm{y} - \bm{X} \bm{\beta} ||_2^2 ) + \mathbb{E}_q( \bm{\beta}^\top \bm{D}^* \bm{\beta} ) + 2d}{2}, \\
\end{array}
\end{equation}
where $(a, b)$ are estimated from MML and $c=d=10^{-5}$. From (\ref{MFVBeqns})-(\ref{MFVBeqns2}), we can easily construct our coordinate ascent updates and compute the expectations, $\mathbb{E}_{q_2^*}(\lambda_i^{-2})$, $\mathbb{E}_{q_3^*}(\xi_i^{-2})$, $\mathbb{E}_{q_4^*} (\sigma^2)$, and $\mathbb{E}_{q_4^*} (\sigma^{-2})$ in closed form, using standard properties of the $\mathcal{GIG}$ and $\mathcal{IG}$ densities. Moreover, we also have
\begin{equation*}
\begin{array}{l}
\mathbb{E}_{q_1^*} ( \beta_i^2 ) = (\beta_i^*)^2 + \bm{\Sigma}^{*}_{ii}, \\
\mathbb{E}_{q_1^*} ( || \bm{y} - \bm{X} \bm{\beta} ||_2^2 ) = || \bm{y} - \bm{X} \bm{\beta}^* ||_2^2 + \textrm{tr} (  \bm{X}^\top \bm{X} \bm{\Sigma}^* ), \\
\mathbb{E}_q (\bm{\beta}^\top \bm{D}^* \bm{\beta}) = \displaystyle \sum_{i=1}^{p} (\beta_i^*)^2  \mathbb{E}_{q_2^*} (\lambda_i^{-2}) \mathbb{E}_{q_3^*} (\xi_i^{-2}) + \textrm{tr} ( \bm{D}^* \bm{\Sigma}^* ). 
\end{array}
\end{equation*}
At each iteration, we compute the evidence lower bound (ELBO),
\begin{equation} \label{ELBO}
\mathcal{L} = \mathbb{E}_q \log f(\bm{y}, \bm{\beta}, \bm{\lambda}, \bm{\xi}, \sigma^2) - \mathbb{E}_q \log q(\bm{\beta}, \bm{\lambda}, \bm{\xi}, \sigma^2),
\end{equation}
where $f$ is the joint density over $\bm{y}$ and all parameters and the expectations in (\ref{ELBO}) are taken with respect to the density $q$ in (\ref{MFVB}). In particular, (\ref{ELBO}) can be found by solving
\begin{equation*} 
\begin{array}{ll}
\mathcal{L} = & \mathbb{E}_{q} \log f(\bm{y} | \bm{\beta}, \sigma^2) + \mathbb{E}_{q} \log ( \bm{\beta} | \bm{\lambda}, \bm{\xi}, \sigma^2 ) + \mathbb{E}_{q} \log \pi(\bm{\xi}) + \mathbb{E}_{q} \log \pi(\bm{\xi})  \\
& + \mathbb{E}_{q} \log \pi(\sigma^2) - \mathbb{E}_{q} \log q_1^* (\bm{\beta}) - \mathbb{E}_{q} \log q_2^* (\bm{\lambda}) - \mathbb{E}_q \log q_3^* (\bm{\xi}) - \mathbb{E}_q \log q_4^*(\sigma^2).
\end{array}
\end{equation*}
Although a bit involved, an explicit expression for the ELBO can be derived in closed form. Namely, we have
\begin{align*} \label{ELBO2}
& \mathcal{L} = -\frac{n}{2} \log (2 \pi) + \frac{p}{2} + p \log 2 + p \log \Gamma(u^*) - p \log \Gamma(a) - p \log \Gamma(b) \\
& \qquad \qquad+ c \log d  - c^* \log d^* + \log \Gamma(c^*) - \log \Gamma(c) + \frac{1}{2} \log | \bm{\Sigma}^* |  \\
& \qquad \qquad - \displaystyle \sum_{i=1}^{p} \log \left[ \frac{(k_i^* / l^* )^{m^*/2}}{K_{m^*} ( \sqrt{k_i^* l^* })} \right] - u^* \displaystyle \sum_{i=1}^{p} \log v_i^* + \displaystyle \sum_{i=1}^{p} \left( \frac{k_i^*}{2}-1 \right) \mathbb{E}_{q_2^*} ( \lambda_i^2)  \\
& \qquad \qquad   + \displaystyle \sum_{i=1}^{p} (v_i^* - 1) \mathbb{E}_{q_3^*} (\xi_i^{-2}) + \frac{l^*}{2} \displaystyle \sum_{i=1}^{p} \mathbb{E}_{q_2^*}(\lambda_i^{-2}),   \numbereqn
\end{align*}
where $K_{\nu}(\cdot)$ denotes the modified Bessel function of the second kind. 

 In each step of our algorithm, we compute the ELBO (\ref{ELBO}). Convergence is assessed by computing the absolute difference, $\textrm{dif} = |\mathcal{L}^{(t)} - \mathcal{L}^{(t-1)} |$, at each iteration, and terminating the algorithm if $\textrm{dif} < \delta$, for some small tolerance $\delta > 0$. We run the MFVB algorithm until convergence or until a maximum of 1000 iterations have been reached.

To incorporate the EM algorithm for computing hyperparameters $(a, b)$ into the MFVB scheme, we solve for $(a, b)$ in (\ref{Mstep}) in every iteration of coordinate ascent algorithm, using $\mathbb{E}_{q_2^{* (t-1)}, a^{(t-1)}} \left[ \log (\lambda_i^2 ) \right]$ and $\mathbb{E}_{q_3^{* (t-1)}, b^{(t-1)}} \left[ \log (\xi_i^2) \right]$ in place of the summands in (\ref{Mstep}) at the $t$th iteration. Namely, these expectations are given by:
\begin{equation} \label{MFVBexpectations} 
\begin{array}{l}
\mathbb{E}_{q_2^{* (t-1)}, a^{(t-1)}} \left[ \log (\lambda_i^2 ) \right] = \log \left( \frac{\sqrt{k_i^{*(t-1)}}}{\sqrt{l^*}} \right) + \frac{\partial}{\partial m^{*(t-1)}} \log \left[ K_{m^{*(t-1)}} \left( \sqrt{ k_i^{*(t-1)} l^{*} } \right) \right], \\
\mathbb{E}_{q_3^{* (t-1)}, b^{(t-1)}} \left[ \log (\xi_i^2 ) \right] = \log \left( v_i^{*(t-1)} \right) - \psi \left( u^{*(t-1)} \right), \\
\end{array}
\end{equation}
where $K_{\nu} ( \cdot )$ denotes the modified Bessel function of the second kind, and $a^{(*t-1)}$, $b_i^{*(t-1)}$, $k_i^{*(t-1)}$, $l^*$, and $m^{*(t-1)}$ are taken from the $(t-1)$st iteration and defined in (\ref{MFVBeqns2}). Numerical differentiation is used to evaluate the derivative in the first equation of (\ref{MFVBexpectations}). Algorithm 2 at the end of Appendix \ref{AppB} provides the complete steps for implementing the variational EM algorithm for the self-adaptive NBP model.

Note that Step 9 in Algorithm 2 involves computing the inverse of a $p \times p$ matrix, $\bm{\Phi}^{*(t)} = (\bm{X}^\top \bm{X} + \bm{D}^{*(t)})^{-1}$. Since $\bm{D}^{*(t)}$ is a diagonal matrix, the computational cost can be substantially reduced when $p \gg n$ by invoking the Sherman-Morrison-Woodbury formula, i.e.
\begin{align*}
\bm{\Phi}^{*(t)} \gets ( \bm{D}^{*(t)} )^{-1} - ( \bm{D}^{*(t)} )^{-1} \bm{X}^\top ( \bm{I}_n + \bm{X} ( \bm{D}^{*(t)} )^{-1} \bm{X}^\top )^{-1} \bm{X} ( \bm{D}^{*(t)} )^{-1},
\end{align*}
which only involves inverting an $n \times n$ matrix, rather than $p \times p$ one. In steps 12-14 of Algorithm 2, we can also update $(k_i^{*(t)}, v_i^{*(t)}), i =  1, \ldots, p$, simultaneously in parallel to save on computing time. 

\begin{algorithm}[H]
  \caption{Monte Carlo EM algorithm for the self-adaptive NBP}
  \begin{algorithmic}[1]
    \Initialize{}
 \vspace{-.2cm}
   \State $a^{(0)} = b^{(0)} = 0.01, c = d = 10^{-5}$,  max = 100, $M = 100, J = 20000$, $\delta = 10^{-6}$, dif = 1, and $k = 0$. 
\State  Initialize $\bm{\beta}^{(0)}, \sigma^{2(0)}, \lambda_i^{2(0)}, \xi_i^{2(0)}, i = 1, \ldots, p.$
      \For{$t = 1$ to $J$}
        \State $\bm{D}^{(t)} \gets \textrm{diag} \left( \lambda_1^{2(t-1)} \xi_1^{2(t-1)}, \ldots, \lambda_p^{2(t-1)} \xi_p^{2(t-1)} \right) $
     \State Draw $\bm{\beta}^{(t)} \sim \mathcal{N}_p \left(\left( \bm{X}^\top \bm{X} + (\bm{D}^{(t)})^{-1} \right)^{-1} \bm{X}^\top \bm{y}, \sigma^{2(t-1)} \left( \bm{X}^\top \bm{X} + (\bm{D}^{(t)})^{-1} \right)^{-1}  \right)$
     \For{$i = 1$ to $p$}
    	\State Draw $\lambda_i^{2(t)} \sim \mathcal{IG} \left( a^{(k)}+ \frac{1}{2}, \frac{ \left( \beta_i^{(t)} \right)^{2}}{2 \sigma^{2(t-1)} \xi_i^{2(t-1)}}+1 \right)$
         \State  Draw $\xi_i^{2(t)} \sim \mathcal{GIG} \left( \frac{ \left( \beta_i^{(t)} \right)^2}{\sigma^{2(t-1)} \lambda_i^{2(t)}}, 2, b^{(k)} - \frac{1}{2} \right)$
      \EndFor
      \State Draw $\sigma^{2(t)} \sim \mathcal{IG} \left(\frac{n+p+2c}{2}, \frac{ ||\bm{y} - ||\bm{X} \bm{\beta}^{(t)} ||_2^2 + \left( \bm{\beta}^{(t)} \right)^\top \left( \bm{D}^{(t)} \right)^{-1} \bm{\beta}^{(t)} + 2d }{2}  \right)$
    \State \texttt{Update hyperparameters in (\ref{Mstep}).}
   \If{ $ t \textrm{ mod } M = 0$ \textbf{and} $k \leq$ max \textbf{and} dif $\geq \delta$  }
       \State $k \gets k+1$
        \State low $\gets t-M+1$
        \State high $\gets t$
      \For{$j = 1$ to $p$}
        \State $U_j\gets \frac{1}{M} \left[ \ln \left( \lambda_j^{2 \textrm{(low)}} \right) + \ldots + \ln \left( \xi_j^{2\textrm{(high)}} \right) \right]$
       \State $V_j \gets \frac{1}{M} \left[ \ln \left( \lambda_j^{\textrm{2(low)}} \right) + \ldots + \ln \left( \xi_j^{\textrm{2(high)}} \right) \right]$
   \EndFor
     \State Solve for $a$ in $-p \psi(a) - \sum_{j=1}^{p} U_j = 0$
     \State  $a^{(k)} \gets a$
     \State Solve for $b$ in $-p \psi(b) + \sum_{j=1}^{p} V_j = 0$
     \State $b^{(k)} \gets b$
     \State dif $\gets \left( a^{(k)} - a^{(k-1)} \right)^2 + \left( b^{(k)} - b^{(k-1)} \right)^2$
   \EndIf

    \EndFor
  \end{algorithmic}
\end{algorithm} 

\begin{algorithm}[H]
	\caption{Variational EM algorithm for the self-adaptive NBP model}
	\begin{algorithmic}[1]
		\Initialize{}
		\vspace{-.3cm}
		\State $l^* = 2$, $c^* = \frac{n+p+2c}{2}$, $a^{(0)} = b^{(0)} = 0.01$, $\delta = 10^{-3}$, $J = 1000$, \textbf{and} $t = 1$.
\State Initialize $d^{*(0)}$, $k_i^{*(0)}, v_i^{*(0)}$, $i=1,\ldots, p$. 
		\vspace{.5cm}
		\While{$ |\mathcal{L}^{(t)} - \mathcal{L}^{(t-1)} | \geq \delta$ \textbf{and} $1 \leq t \leq J$}
		\State \texttt{E-step: Update variational parameters in (\ref{MFVBeqns2}).}
		\State Update $m^{*(t)} \gets a^{(t-1)} - \frac{1}{2}$
		\State Update $u^{*(t)} \gets b^{(t-1)} + \frac{1}{2}$
		\State Update $\bm{D}^{*(t)} \gets \textrm{diag} \left( \mathbb{E}_{q_2^{*(t-1)}} (\lambda_1^{-2}) \mathbb{E}_{q_3^{*(t-1)}} ( \xi_1^{-2}), \ldots, \mathbb{E}_{q_2^{*(t-1)}} (\lambda_p^{-2}) \mathbb{E}_{q_3^{*(t-1)}} (\xi_p^{-2})  \right)$
		\State Update $\bm{\Phi}^{*(t)} \gets \left( \bm{X}^\top \bm{X} + \bm{D}^{*(t)} \right)^{-1}$
		\State Update $\bm{\Sigma}^{*(t)} \gets \mathbb{E}_{q_4^{*(t-1)}} (\sigma^2) \bm{\Phi}^{*(t)}  $
		\State Update $\bm{\beta}^{*(t)} \gets \bm{\Phi}^{*(t)} \bm{X}^\top \bm{y}$
		\For{$i = 1$ to $p$}
		\State Update $k_i^{*(t)} \gets \mathbb{E}_{q_1^{*(t-1)}} ( \beta_i^{2} ) \mathbb{E}_{q_4^{*(t-1)}} (\sigma^{-2}) \mathbb{E}_{q_3^{*(t-1)}} ( \xi_i^{-2} )$
		\State Update $v_i^{*(t)} \gets \frac{1}{2} \mathbb{E}_{q_1^{*(t-1)}} ( \beta_i^{2} ) \mathbb{E}_{q_4^{*(t-1)}} (\sigma^{-2}) \mathbb{E}_{q_2^{*(t-1)}} (\lambda_i^{-2} ) + 1$
		\EndFor
		\State Update $d^{*(t)} \gets \frac{ \mathbb{E}_{q_1^{*(t-1)}} \left( || \bm{y} - \bm{X} \bm{\beta}^2 ||_2^2 \right) + \mathbb{E}_{q^{(t-1)}} \left(  \bm{\beta}^\top \bm{D}^{*(t)} \bm{\beta} \right) + 2d}{2}$
		\vspace{.5cm}
		\State \texttt{M:step: Update hyperparameters in (\ref{Mstep}).}
		\State Solve for $a$ in $-p \psi(a) +  \sum_{i=1}^{p} \mathbb{E}_{q_2^{*(t-1)}} \left[ \log (\lambda_i^2) \right] = 0$
		\State  $a^{(t)} \gets a$
		\State Solve for $b$ in $-p \psi(b) -  \sum_{i=1}^{p} \mathbb{E}_{q_3^{*(t-1)}} \left[ \log (\xi_i^2) \right] = 0$
		\State  $b^{(t)} \gets b$
		\State Update $\mathcal{L}^{(t)}$, as in (\ref{ELBO2}).
		\State $t \gets t+1$     
		\EndWhile
	\end{algorithmic}
\end{algorithm}

\end{appendix}

\end{document}